\documentclass[twocolumn]{autart}

\newcommand{\eE}{\varepsilon_E}

\newcommand{\eT}{\varepsilon_T}

\newcommand{\fhat}{\hat{f}}
\newcommand{\ygoal}{y_{\mathrm{goal}}}

\newcommand{\Bgam}{\mathcal{B}_\gamma}
\newcommand{\yhat}{\hat{y}}
\newcommand{\zhat}{\hat{z}}

\newcommand{\Yhat}{\hat{\mathcal{Y}}}
\newcommand{\Yf}{\mathcal{Y}_f}
\newcommand{\YN}{\mathcal{Y}_N}

\newcommand{\Vbar}{\bar{V}_h}

\usepackage{tikz}
\usetikzlibrary{positioning}
\usepackage{tcolorbox}
\usepackage{amsmath}
\usepackage{amssymb}
\usepackage{empheq}
\usepackage[font={small}]{caption}
\usepackage{graphicx}
\usepackage{xcolor}
\usepackage{comment}
\usepackage{url}
\usepackage{subcaption}
\usepackage{algorithm}
\usepackage[noend]{algpseudocode}
\usepackage{booktabs}
\usepackage{mathtools}
\usepackage{multirow}

\usepackage{enumerate}

\newtheorem{theorem}{Theorem}
\newtheorem{definition}{Definition}
\newtheorem{corollary}{Corollary}
\newtheorem{proposition}{Proposition}

\newtheorem{remark}{Remark}

\newtheorem{assumption}{Assumption}
\newtheorem{lemma}{Lemma}

\newcounter{num}

\newenvironment{proof}{\noindent {\it Proof.}}{ \hfill \mbox{\footnotesize $\blacksquare$ } \\ }

\usepackage{tabularx}

\usepackage{cite}

\newif\ifextended

\ifextended
  \newcommand{\appref}[1]{Appendix~\ref{#1}}
  \newcommand{\Appref}[1]{Appendix~\ref{#1}}
\else
  \newcommand{\appref}[1]{the appendix of the extended version~\cite{myextended}}
  \newcommand{\Appref}[1]{The appendix of the extended version~\cite{myextended}}
\fi

\begin{document}

\begin{frontmatter}
\title{Safety by Invariance, Liveness through Refinement: \\ Heterogeneous Contract Framework for \\ Co-Design of Layered Control}

\thanks[footnoteinfo]{This work has received funding from the Agence Nationale de la Recherche (ANR) under the ESTHER grant ANR-22-CE05-0016, from "France 2030" program through the TASTING grant ANR-22-PETA-0012, and from the RTE-CentraleSupélec chair. Corresponding Author: Y. Takayama. \textit{Emails:} \{yoshinari.takayama, alessio.iovine, guillaume.sandou\}@centralesupelec.fr, \{b.besselink\}@rug.nl, \{adnane.saoud\}@um6p.ma}
\author[First,Second,Third]{Yoshinari Takayama}
\author[First]{Alessio Iovine}
\author[Second]{Bart Besselink}
\author[First]{Guillaume Sandou}
\author[Third]{Adnane Saoud}

\address[First]{Laboratory of Signals and Systems (L2S),  CNRS, CentraleSupelec, Paris-Saclay University, France}
\address[Second]{Bernoulli Institute for Mathematics, Computer Science, and
Artificial Intelligence, University of Groningen, The Netherlands}
\address[Third]{College of Computing, University Mohammed VI Polytechnic, Benguerir, Morocco}

\begin{abstract}
Real-world control systems must achieve complex, long-horizon objectives (liveness) while strictly respecting safety constraints in continuous time. Addressing both objectives simultaneously poses significant challenges for either a discrete-time planner, such as Model Predictive Control (MPC), or a continuous-time tracking controller acting alone, motivating hierarchical layered control architectures (LCAs). Existing LCA research, however, lacks generality owing to three open difficulties: (i) the absence of a uniform specification language capable of expressing requirements across both discrete planning and continuous execution; (ii) the lack of formal guarantees that these specifications are preserved when interconnecting subsystems operating at heterogeneous time scales; and (iii) reliance on naive input-filtering laws that obstruct compositional separation across layers. This paper addresses all three gaps by importing the safety–liveness decomposition from computer science into a heterogeneous assume-guarantee framework: \emph{safety is enforced by invariance} at the continuous-time layer, while \emph{liveness is achieved through refinement} at the discrete-time layer. Inter-layer coordination is formalized via vertical refinement and timing compatibility conditions. We instantiate this contract with a novel LCA that connects an MPC, an input-to-state stabilizing (ISS) low-level controller, and a reference-governor mechanism that bridges them, stating explicit conditions on each component. The resulting implementation is validated on a Hybrid Energy Storage System (HESS) comprising a battery and a supercapacitor.
\end{abstract}
\end{frontmatter}
\newpage
\section{Introduction}

Real-world control systems face a dual challenge: achieving complex, long-horizon goals (such as those expressed in temporal logics~\cite{Fainekos2009-yn,LindemannBook, STL-GO, yoshinari2023jccpounal}) while strictly adhering to safety constraints in continuous time.
In industrial practice, this hierarchy is naturally addressed by coupling a high-level system, such as Model Predictive Control (MPC)~\cite{Borrelli2017-rp,Rawlings2017-vw}, with low-level tracking controllers, e.g., Proportional Integral Derivative (PID) ones. 
While this combination is intuitive, formally verifying the correctness of such composite systems remains a significant theoretical hurdle, due to the mismatch in time domains (discrete-time planning vs. continuous-time physics). Unlike hybrid systems, where continuous and discrete dynamics coexist within a single dynamical model, a layered architecture~\cite{LCAMatni, Incer_2024} separates the system into distinct functional layers — a discrete-time planner that computes reference commands and a continuous-time controller that tracks them — each designed and analyzed independently, interacting only through well-defined interfaces.
This separation raises the central question of this paper: how can a low-level controller be certified to maintain safety in continuous time while the planner's discrete-time guarantees are certified separately, so that the two can be recombined
across the time-domain mismatch---without a circular cross-layer dependence, and
without introducing deadlocks or invalidating the planner's prediction model?

Standard approaches often rely on heuristic time-scale separation arguments \cite{doflerTSS} or employ Control Barrier Functions (CBFs)~\cite{Ames2019-rq,Yuan18052024} to enforce safety.
A QP-based CBF filter can ignore inter-sample behavior~\cite{AuxiVarBelta}, but the more subtle difficulty is compositional rather than numerical: because such a filter acts on the control input $u$, it alters the very closed-loop tracking behavior that the planner's prediction model presupposes~\cite{rosolia2020multi,garg2021multi,RosoliaUnified}. The planner's assumption about what the low level delivers and the low level's realized guarantee are then coupled through $u$, so the two layers can no longer be analyzed independently.
This motivates moving the safety intervention from the input to the \emph{reference},
where it leaves the inner-loop certificates---and hence the planner's
assumption---intact. Conversely, while rigorous verification frameworks like vertical
abstraction exist~\cite{mazoLCA,Nuzzo2018}, establishing simulation
relations~\cite{Tabuada_undated-oi} often requires exhaustive reachability analysis,
which can be computationally prohibitive for complex dynamics.

On the other hand, assume-guarantee (AG) contract theory provides a powerful
semantic framework: under an \emph{assumption} on inputs or environment
behavior, each subsystem must provide a \emph{guarantee}. Originally developed
in computer science~\cite{benveniste2018contracts, nuzzo2015platform,
hasuoerato, muller2016ifm} with later algebraic
extensions~\cite{SHARF2024111637, alaoui2024adhs, incer2022hypercontracts} and
cyber-physical applications such as aircraft electric power
systems~\cite{nuzzo2014contract}, it established the core operations of
composition, refinement, and conjunction at a single level of abstraction
(\emph{horizontal} contracts). More recently, the control community has tailored
AG contracts to continuous-state dynamical systems through a range of semantic
and computational formulations---parametric contracts linked to small-gain
theory~\cite{kim2017smallgain}, weak/strong semantics for continuous- and
discrete-time systems~\cite{saoud2021assume, saoud2021tac} with STL
extensions~\cite{SiyuanContract} and resilience-based feasibility~\cite{monir2025cdc},
and behavioral formulations for linear systems~\cite{shali2023tac, shaliDC,
ArminHierarchi}. Complementary lines address invariant-set
verification~\cite{girard2022cdc}, convex compositional
synthesis~\cite{ghasemi2024automatica}, and applications such as DC-microgrid
voltage regulation~\cite{zonetti2019ecc}.

Beyond horizontal contracts, contract theory has been extended to mixed discrete-continuous and layered architectures \cite{LCAMatni}. 
However, existing
frameworks struggle to capture the feedback structures specific to layered control,
where components operate across multiple timescales and heterogeneous signal spaces,
and the constructive mechanisms proposed to address
this~\cite{nuzzoLCA,mazoLCA,Incer_2024,nuzzo2014contract,GirardApprox} have focused
primarily on the \emph{verification} of given designs. The closest prior work either
verifies layered designs within a contract algebra~\cite{mazoLCA,Nuzzo2018} or
synthesizes multirate safe controllers at the \emph{input} level, without an
assume-guarantee account of the cross-layer
obligations~\cite{rosolia2020multi,garg2021multi,RosoliaUnified}. In contrast, we
contribute a \emph{vertical} contract whose handshakes are deliberately
asymmetric---safety refines downward unilaterally, liveness refines
bilaterally---and which is \emph{realized} by a reference-level governor that
preserves the inner-loop certificates the planner relies on. To our knowledge this is
the first layered contract in which safety and liveness are discharged by structurally
distinct mechanisms and recombined without a circular cross-layer dependence.

Co-design frameworks have addressed the joint optimization of interdependent system components, from hardware--software stacks for autonomous systems using monotone and category-theoretic methods \cite{ZardiniEtAl2021IROS, zardini2023camod} to planner--controller synthesis under temporal logic specifications \cite{pant2020co}. However, these approaches typically optimize over a shared design space without formalizing explicit assume-guarantee obligations between layers, and do not address the synthesis of controllers that jointly enforce safety and liveness within a layered architecture with heterogeneous timescales and signal spaces.

This paper makes the following contributions.

\begin{itemize}
\item \textbf{A directional vertical contract for heterogeneous-time layers.}
We compose a discrete-time (DT) planning contract $\mathcal C_H$ with a
continuous-time (CT) safety contract $\mathcal C_L$ across two time domains,
coordinated by two \emph{asymmetric} handshakes (Defs.~\ref{def:vertical_refinement}--\ref{def:heterogeneous_wellposed}):
a downward refinement of the reference assumption and an upward refinement of the
tracking guarantee into the planner's mismatch assumption. The asymmetry is deliberate.
Safety is \emph{unilateral}---the CT layer enforces it by robust forward invariance
regardless of what the planner commands, even under the frozen fallback
$r_k=r_{k-1}$, since
the low-level layer enforces invariance regardless of $r$---while liveness is \emph{bilateral},
requiring both handshakes and a timing condition. By modeling the zero-order hold
explicitly, the inter-layer dependence is made sequential, so the composite is
well-posed by construction and correctness follows by induction over sampling
intervals (Thm.~\ref{thm:layered_correctness}), with no algebraic loop to break.

\item \textbf{An explicit, computable scale-separation condition.}
In place of asymptotic singular-perturbation arguments~\cite{Kokotovic1986}, we give a
finite-time timing condition $C_{\mathrm{tss}}:T_s\ge\tau_{LL}$
(Def.~\ref{def:timing_compatibility}) with an explicit estimate of the settling time $\tau_{LL}$
(Prop.~\ref{prop:settling}), linking the low-level decay rate directly to the
high-level sampling period. The bridge from DT MPC convergence to CT contract
satisfaction is an input-to-state-stability argument on the MPC value function
(Lemma~\ref{lem:dissipation}, Prop.~\ref{prop:mpc_iss}): model--reality mismatch enters as a
bounded ISS disturbance with input radius $\varepsilon_E$ and output radius
$\varepsilon_T$, so the planner's convergence guarantee becomes a quantity the CT
layer can supply. This provides a constructive bridge from discrete-time MPC to continuous-time assume-guarantee correctness.

\item \textbf{The explicit reference governor as a contract realizer.}
We recast the explicit reference governor
(ERG)~\cite{GaroneNicotraERG,NicotraUAVERG}---an existing constrained-control
mechanism---as the object that \emph{realizes} the vertical contract. Acting on the
reference rather than the input, it simultaneously enforces $\varphi_{\mathrm{safe}}$
by robust forward invariance (Thm.~\ref{thm:erg_invariance}) and supplies the tracking guarantee
that closes the upward handshake (Cor.~\ref{cor:lowlevel}), while leaving the deployed
inner-loop controller and its stability/tracking certificates untouched. This is what
lets a legacy inner loop be retained without redesign under the contract: the MPC adds
high-level planning, and the compositional framework supplies the formal guarantees.

\item \textbf{Instantiation on a hybrid energy storage system (HESS).}
We instantiate the framework on a battery--supercapacitor HESS whose slow/fast split
realizes the safety--liveness decomposition \emph{at the component level}: the
supercapacitor rejects fast disturbances (safety), the battery tracks the slow energy
target (liveness) (Remark~\ref{rem:superbatt}). A feedforward design removes the known
load from the error dynamics so that the invariance machinery handles only the unknown
residual; numerical results confirm safety and liveness under an off-target initialization and bounded disturbances. The timing condition is sufficient and, as the case study
shows, conservative when strong feedforward is available.
\end{itemize}

\textbf{Organization.}
Section~\ref{sec:preliminaries} formalizes the problem of layered controller synthesis with safety and liveness specifications. Section~\ref{sec:heterogeneous_contracts} develops a contract framework based on vertical refinement to ensure correctness across heterogeneous time domains and bridge the abstraction gap between the planner's simplified model $\hat{f}$ and the physical dynamics $f$. Section~\ref{sec:lca} implements the proposed contract framework using the ERG-tracker and MPC combination. The framework is applied to a HESS in Section~\ref{sec:HESS}.
Section~\ref{sec:simulation} presents numerical simulations validating the results. Section~\ref{sec:conclusion} concludes the paper.

\paragraph*{Notation.}
We denote by $\mathbb{T}$ the time index set:
$\mathbb{T} = \mathbb{R}_{\geq 0}$ in continuous time
and $\mathbb{T} = \mathbb{N}$ in discrete time.
In the latter case, the step index $k \in \mathbb{N}$ corresponds to physical time
$t_k = k T_s$, where $T_s > 0$ is the sampling period.
For a set $\mathcal{S}$, the \emph{signal space} for a signal
$\xi$ on $\mathbb{T}$ is $\mathcal{S}^{\mathbb{T}} := \{\, \xi \colon \mathbb{T} \to \mathcal{S} \,\}$ denotes the set of all $\mathcal{S}$-valued trajectories indexed by $\mathbb{T}$.
A component $\Sigma$ is identified with its trajectory set (behavior)
$\mathcal{B} \subseteq \mathcal{S}^{\mathbb{T}}$.
For $x \in \mathbb{R}^n$, $\|x\|$ denotes the Euclidean norm; for a signal
$\xi$ on $\mathbb{T}$, $\|\xi\|_{\infty} := \operatorname*{ess\,sup}_{t \in \mathbb{T}} \|\xi(t)\|$.
A function $\alpha: \mathbb{R}_{\ge 0} \to \mathbb{R}_{\ge 0}$ is class $\mathcal{K}$ 
if continuous, strictly increasing, and $\alpha(0)=0$; class $\mathcal{K}_\infty$ 
if additionally $\alpha(r) \to \infty$ as $r \to \infty$. A function 
$\beta: \mathbb{R}_{\ge 0}^2 \to \mathbb{R}_{\ge 0}$ is class $\mathcal{KL}$ if 
$\beta(\cdot,t) \in \mathcal{K}$ for fixed $t$ and $\beta(r,\cdot)$ decreases 
to zero for fixed $r$. 
\section{Problem Formulation}
\label{sec:preliminaries}

\subsection{System Models}

We consider a two-layer control architecture comprising a discrete-time (DT) 
planning layer and a continuous-time (CT) safety layer (cf. Fig. \ref{fig:architecture}). The planning layer operates at a sampling period $T_s$, generating references 
based on an abstracted model. The input to the planning layer is the sampled state $y_k \in \mathbb{R}^q$, 
and the output is a reference $r_k \in \mathbb{R}^p$. The safety layer is implemented on the continuous-time model.

\begin{figure}[ht!]
    \centering \includegraphics[width=0.46\textwidth]{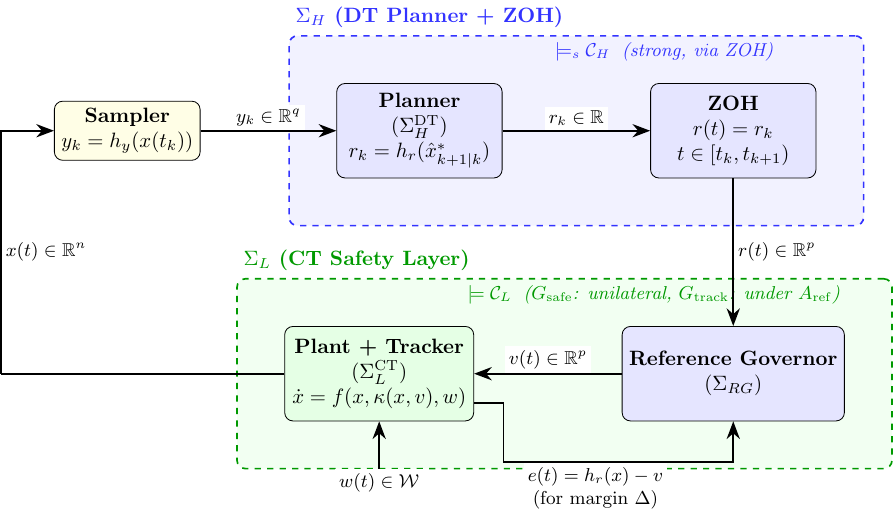}
    \caption{Layered control architecture. Signals: $r_k$ (reference), 
    $y_k$ (sampled state), $v(t)$ (filtered reference), $u(t)$ (control), 
    $x(t)$ (state), $w(t)$ (disturbance). Signal flow: $r \to v \to x$.}
    \label{fig:architecture}
\end{figure}

\begin{definition}[Low-Level CT System]
\label{def:ct_plant}
The low-level continuous-time system $\Sigma_L^{\mathrm{CT}}$ is defined by
\begin{align}\label{eq:ct_system}
    \dot{x}(t) = f(x(t),\, u(t),\, w(t)),
\end{align}
where $x(t) \in \mathcal{X} \subseteq \mathbb{R}^n$ 
is the state, $u(t) \in \mathcal{U} \subseteq 
\mathbb{R}^m$ is the control input, and 
$w(t) \in \mathcal{W} \subseteq \bigl\{w \in \mathbb{R}^d : \|w\| \leq W_{\max}\bigr\}$
represents external disturbances. The tracked 
output is $h_r(x) = P_r x \in \mathbb{R}^p$, 
where $P_r \in \mathbb{R}^{p \times n}$ is a 
selection matrix. The low-level tracking controller generates:
\begin{align}
    u(t) = \kappa(x(t),\, v(t)),
\label{eq:tracking_controller}
\end{align}
where $v(t) \in \mathbb{R}^p$ is the reference signal
applied to the tracker that follows a dynamics that will be described in Section \ref{sec:lca}. This filtered reference $v$ is a filtered version of the high-level system's output $r(t)$. The vector field $f$ 
is locally Lipschitz in $x$, $\kappa$ is locally 
Lipschitz, $r(t)$ is piecewise constant, and 
$w(t) \in \mathcal{W}$ is piecewise continuous. 
If every solution remains in a compact subset of 
$\mathcal{X}$, then a unique solution exists for 
all $t \geq t_0$~\cite{Khalil}. The input-output trajectory set of the low-level system $\Sigma_L$ is
{\scriptsize 
\begin{align}
    \mathcal{B}_L := \{(w, r, x) : 
    \mathbb{R}_{\geq 0} \to 
    \mathcal{W} 
    \times \mathbb{R}^p \times \mathcal{X} \mid 
    \text{\eqref{eq:ct_system}, 
    \eqref{eq:tracking_controller} hold a.e.}\}.
\end{align} 
}
\end{definition}
Note that the low-level layer $\Sigma_L$ receives the raw reference $r\in \mathbb{R}^p$ rather than the filtered reference $v\in \mathbb{R}^p$, since $\Sigma_L$ encompasses both the plant–tracker block and the reference governor (cf. Fig. \ref{fig:architecture})\footnote{When $v$ is absolutely continuous—as the ERG (Section \ref{sec:erg})  ensures—equations \eqref{eq:ct_system}–\eqref{eq:tracking_controller} hold everywhere; the 'a.e.' qualifier covers only $v=r$ (piecewise constant), where the ODE holds in the Carathéodory sense~\cite{AubinCellina2012, Goebel2012-hs}.}.

To define the input--output trajectory set of the high-level system, we first introduce the abstract discrete-time model of the plant on which the planner is based:
\begin{align}
    \yhat_{k+1} = \fhat(\yhat_k,\, \zhat_k,\, r_k),
    \label{eq:abstract_model}
\end{align}
where the map $\fhat$ is a (discrete-time) ``abstraction" of the original dynamics $f$. 
Whereas the dynamics $f$ governs the full state $x \in \mathbb{R}^n$,
the reduced dynamics $\fhat$ governs the abstract state $\yhat \in \mathbb{R}^q$
with $q \leq n$; the signal $\zhat \in \mathbb{R}^{n-q}$ collects the remaining components of the state that are not retained in the abstract model. 

Consider a finite receding horizon of length $N$ steps starting from
time~$k$. Let $\yhat_{k+j|k}$ denote the $j$-step-ahead predicted output, conditioned on information available at time $k$. The open-loop predicted trajectory is then generated by iterating the model in \eqref{eq:abstract_model}: for all $j \in [ 0, \dots, N{-}1],$
\begin{align}
    \yhat_{k+j+1|k} = \fhat(\yhat_{k+j|k},\,
    \zhat_{k+j|k},\, r_{k+j}),
    \label{eq:predicted_model}
\end{align}
with the initial condition $\yhat_{k|k} = y_k := h_y(x(t_k)) \in \mathbb{R}^q$ from the sampled measurement at each step.

The closed-loop high-level DT system $\Sigma_H^{\mathrm{DT}}$ solves a finite-horizon optimal control problem of length $N$ steps in a receding-horizon fashion, reapplied every sampling period $T_s$. When interconnected in a closed loop with the low-level CT system, the mismatch between the predicted and actual dynamics causes the high-level DT system's state to evolve as follows.

\begin{definition}[High-Level DT System]
\label{def:dt_planner}
Under the consistency condition $\yhat_{k|k} = y_k$, the high-level DT system $\Sigma_H^{\mathrm{DT}}$ is described by
\begin{align}
    y_{k+1} &= \hat{f}(y_k, \hat{z}_{k|k}, r_k) + \tilde{w}_{k+1},
    \label{eq:realised_dynamics}\\
    r_k &= \pi(y_k,\, \zhat_{k|k}) \in \mathbb{R}^p,
    \label{eq:reference_extraction}
\end{align}
where the one-step prediction error is indexed to the predicted step,
\begin{equation}
  \tilde{w}_{k} := y_{k}
    - \hat{f}\bigl(\hat{y}_{k-1|k-1}, \hat{z}_{k-1|k-1}, r_{k-1}\bigr),
  \label{eq:wtilde-def}
\end{equation}
so that $\tilde{w}_k \in \mathbb{R}^{q}$ bounds the mismatch in predicting
$y_k$ from the information available at step $k-1$. $\pi$ represents the policy that results from the planner's design. The input--output trajectory set of $\Sigma_H^{\mathrm{DT}}$ is therefore
\begin{align}
    \mathcal{B}_H &:= \Bigl\{\,
    (\tilde{w},\, y,\, r) \colon
    \mathbb{N} \to
\bigl(\mathbb{R}^q \times \mathbb{R}^q \times \mathbb{R}^p\bigr) \;|\; \nonumber\\
    &\text{\eqref{eq:realised_dynamics} and \eqref{eq:reference_extraction} hold with } \yhat_{k|k} = y_k \text{ for all } k \in \mathbb{N}
    \,\Bigr\}.
\end{align}
\end{definition}

Since $\fhat$ assumes perfect tracking, every real-world deviation, such as tracking imperfection, disturbances, and limited model fidelity on untracked channels, is absorbed into $\tilde{w}_k$. In this sense, $\tilde{w}$ enters $\Sigma_H$ as an abstracted disturbance inherited from the low-level CT system.

\begin{definition}[Sample-and-Hold]\label{def:sample_hold}
The interface between the DT and CT systems comprises a \emph{sampler}
$y_k = h_y(x(t_k))$ with $t_k = k T_s$, where $h_y: \mathbb{R}^n \to
\mathbb{R}^q$ extracts the measurable components, and a \emph{zero-order hold}
$r(t) = r_k$ for $t \in [t_k, t_{k+1})$.
\end{definition}

\begin{definition}[Layered Control System]
\label{def:layered_system}
A layered control system 
$\Sigma = \Sigma_H \triangleright \Sigma_L$ is 
the hierarchical interconnection of a high-level 
system $\Sigma_H$ 
(Definition~\ref{def:dt_planner}) and a low-level 
system $\Sigma_L$ (Definition~\ref{def:ct_plant}) 
through the sample-and-hold interface 
(Definition~\ref{def:sample_hold}):
The input-output trajectory set of the composite closed-loop system is
\begin{align}
    \mathcal{B} := \bigl\{(w, x) &:
    \mathbb{R}_{\geq 0} \to \mathcal{W} \times \mathcal{X}
    \;\big|\; \exists\, (r, y_k, \tilde{w}_k) \;\text{s.t.}\; \nonumber \\
    & (w, r, x) \in \mathcal{B}_L \text{ (low level)}, \nonumber \\
    & (\tilde{w}_k, y_k, r_k) \in \mathcal{B}_H \text{ (high level)}, \nonumber \\
    & y_k = h_y(x(t_k))\text{ (sampler)}, \nonumber \\
    & r(t) = r_k \;\;\forall\, t \in [t_k, t_{k+1}) \text{ (hold)},
    \bigr\}.
    \label{eq:layered_traj}
\end{align}
\end{definition}

The information flow is 
sequential: $r(t)$ over $[t_k, t_{k+1})$ depends 
only on $y_k$, so both layers do not require the other's output simultaneously, and the algebraic loop is avoided (see Theorem \ref{thm:layered_correctness}).
\subsection{Temporal Properties in Continuous Time}
\label{subsec:safety_liveness}
The classification of temporal properties into safety and liveness
originates with Lamport~\cite{Lamport1977}.

\begin{description}
    \item[Safety:] ``Something bad never happens'': a property violable
    by a finite prefix of a trajectory; once violated, no extension
    can repair it.

    \item[Liveness:] ``Something good eventually happens'': a property
    no finite prefix can violate; every prefix admits a satisfying
    extension.
\end{description}
Crucially, liveness encompasses not only discrete events ('reach the goal') but also continuous-time convergence, making the framework natural for cyber-physical systems; we treat continuous-time safety as invariance and liveness as convergence\footnote{In the formal methods literature, \emph{invariance} ($\Box \psi$ for state predicate $\psi$) is sometimes distinguished from general \emph{safety} (see Chapter 3.3 of \cite{Baier2008-up}).} and continuous-time liveness as convergence.

\begin{definition}[Safety Specification]
\label{def:safety_spec}
Given a safe set 
$\mathcal{X}_{\mathrm{safe}} \subseteq \mathbb{R}^n$ and a disturbance set $\mathcal{W} \subseteq \mathbb{R}^d$, 
the layered system $\Sigma$ satisfies 
$\varphi_{\mathrm{safe}}$ if for all trajectories $(w, x) \in \mathcal{B}$:
\begin{align}
    x(0) \in \mathcal{X}_{\mathrm{safe}} 
    \implies x(t) \in \mathcal{X}_{\mathrm{safe}} 
    \quad \forall t \geq 0.
    \label{eq:phisafe}
\end{align}
\end{definition}

For simplicity, this paper assumes that the safe set $\mathcal{X}_{\mathrm{safe}}$ is the polyhedron
\begin{equation}\label{eq:safe_set_compact}
    \mathcal{X}_{\mathrm{safe}} := \bigl\{ x \in \mathbb{R}^n \;\big|\; Cx \leq d \bigr\},
\end{equation}
where $C \in \mathbb{R}^{m \times n}$ and $d \in \mathbb{R}^m$. 

For continuous-time systems under persistent 
disturbances, exact convergence is generally 
infeasible. We formalize the liveness property $\varphi_{\mathrm{live}}$ in continuous-time as approximate convergence.

\begin{definition}[$\varepsilon$-Liveness Specification]
\label{def:liveness_spec}
Given a goal $y_{\mathrm{goal}} \in \mathbb{R}^q$ 
and tolerance $\varepsilon > 0$, the layered system $\Sigma$ satisfies the \emph{liveness specification} $\varphi_{\mathrm{live}}$ for all trajectories $(w,x) \in \mathcal{B}$, if 
\begin{align}
    \exists T < \infty, \quad \|h_y(x(t)) - y_{\mathrm{goal}}\|_{\mathrm{live}} 
    \leq \varepsilon \quad \forall\, t \geq T.
    \label{eq:philive}
\end{align}
where $\|\cdot\|_{\mathrm{live}}:=\|P_{\mathrm{live}}\,\cdot\|$ is a seminorm
selecting the regulated coordinate(s). The high-level convergence guarantees are
stated in this seminorm; for the abstract state $\hat y$ we write $\|\cdot\|$ for
$\|\cdot\|_{\mathrm{live}}$ below, the full Euclidean norm being retained for all
other quantities. In the HESS (Section~\ref{sec:HESS}), $P_{\mathrm{live}}$ picks
$E_B$, so $\|\hat y - y_{\mathrm{goal}}\| = |E_B - E_B^{\mathrm{goal}}|$.
\end{definition}

We now formulate the layered controller synthesis problem using the 
system models and specifications above.

\begin{tcolorbox}[colback=purple!5!white,colframe=purple!75, title=Problem 1: Layered Controller Synthesis]
Given the low-level system dynamics~\eqref{eq:ct_system},
constraint sets $\mathcal{X}$, $\mathcal{U}$, $\mathcal{W}$,
safety specification $\varphi_{\mathrm{safe}}$~\eqref{eq:phisafe},
and liveness specification $\varphi_{\mathrm{live}}$~\eqref{eq:philive},
co-design 
\begin{enumerate}[(i)]
    \item the high-level system $\Sigma_H$: a prediction model \eqref{eq:predicted_model}
and finite-horizon optimal control law
\eqref{eq:reference_extraction},
\item the low-level tracking controller
\eqref{eq:tracking_controller} for $\Sigma_L$,
\item a sampling period $T_s > 0$
\end{enumerate}
such that every trajectory
$(w, x) \in \mathcal{B}$ \eqref{eq:layered_traj}
of the layered system
$\Sigma = \Sigma_H \triangleright \Sigma_L$
(Definition~\ref{def:layered_system}) satisfies $\varphi_{\mathrm{safe}}
    \;\wedge\; \varphi_{\mathrm{live}}$.
\end{tcolorbox}


\section{Heterogeneous Contract Framework}
\label{sec:heterogeneous_contracts}

\subsection{Assume-Guarantee Contract}
To solve Problem 1, we adopt the assume-guarantee (AG) contract framework, which provides a compositional means of specifying and verifying the obligations of each layer independently. Rather than analyzing the closed-loop system monolithically, an AG contract assigns to every component a pair of behavioral specifications: an assumption on what the component may legitimately expect from its environment (from outside the network or other components, e.g., references or measurements), and a guarantee on what it must deliver in return. 

\begin{definition}[Assume-Guarantee Contract]
\label{def:ag_contract}
Let $\Sigma$ be a component with trajectory set
$(w,x) \in \mathcal{B} \subseteq \mathcal{S}^{\mathbb{T}}$,
where $w$ is $\Sigma$'s exogenous input, $x$ is the state, and $\mathcal{S} = \mathcal{W} \times \mathcal{X}$. An
AG contract is a pair
$\mathcal{C} = (A, G)$
of sets of trajectories, where $A \subseteq \mathcal{W}^{\mathbb{T}}$ encodes assumptions on the exogenous input and $G \subseteq (\mathcal{W} \times \mathcal{X})^{\mathbb{T}}$ encodes guarantees on the joint input--state behaviour. The component
\emph{satisfies} its contract, if for every trajectory $(w,x) \in \mathcal{B}$:
\begin{align}\label{eq:ag_sat}
    w \in A
    \;\implies\;
    (w,x) \in G,
\end{align}
which is written $\Sigma \models \mathcal{C}$ (we use predicate and set notation
interchangeably: a predicate $P$ defines the set $\{s : P(s)\}$ and vice versa,
favoring whichever is clearer).
\end{definition}

Let us consider co-design of local contracts $\mathcal{C}_L$ and $\mathcal{C}_H$ for each system $\Sigma_L$ and $\Sigma_H$ such that, if $\Sigma_L \models \mathcal{C}_L$ and $\Sigma_H  \models \mathcal{C}_H$ holds, the layered system (Definition~\ref{def:layered_system}) solves Problem 1. 
The challenge is to provide inter-layer interface conditions that ensure
contract correctness across the two time domains while reconciling
$\hat f$ with $f$.

\subsection{Safety-Liveness Decomposition Principle in Vertical Contract Design}

Alpern and Schneider~\cite{Alpern1985} showed that every temporal property decomposes into the intersection of a safety and a
liveness property. We organize the contract around this
decomposition, assigning qualitatively different roles to the two
layers, which is mentioned in the title \textbf{safety by invariance, liveness through refinement}: 

\begin{tcolorbox}[colback=blue!5!white,colframe=blue!75,
    title=Safety-Liveness Decomposition Principle]
\textbf{Safety is unilateral.}\; The low-level enforces
$\varphi_{\mathrm{safe}}$ in continuous time via robust forward
invariance, regardless of what $\Sigma_H$ commands.\\[4pt]
\textbf{Liveness is bilateral.}\; The high-level plans on an
abstract model; the low-level tracks its references; vertical refinement ensures the high-level system's convergence in discrete-time transfers to the physical system in continuous-time.
\end{tcolorbox}


Input-to-state stability can be a common language to state each layer's contribution, applying to both continuous- and discrete-time systems \cite{Khalil, JiangWang2001}.

\begin{definition}[Input-to-State Stability]
\label{def:iss}
A dynamical system with trajectory set
of $(w, \xi)$, state signal
$\xi : \mathbb{T} \to \mathbb{R}^n$, and
disturbance signal
$w : \mathbb{T} \to \mathbb{R}^d$
(where $\mathbb{T} = \mathbb{R}_{\geq 0}$ or
$\mathbb{T} = \mathbb{N}$) is
\emph{input-to-state stable} (ISS) if there
exist $\beta \in \mathcal{KL}$ and
$\gamma \in \mathcal{K}$ such that every
trajectory in $\mathcal{B}$ satisfies:
\begin{align}
    \|\xi(t)\| \leq \beta(\|\xi(0)\|,\, t)
    + \gamma\bigl(\|w\|_\infty\bigr)
    \quad \forall\, t \in \mathbb{T}.
    \label{eq:iss}
\end{align}
The function $\gamma$ is the \emph{ISS gain}
and $\gamma(\|w\|_\infty)$ the
\emph{ultimate bound}: for any $\delta > 0$
there exists finite $T \in \mathbb{T}$ such that
$\|\xi(t)\| \leq \gamma(\|w\|_\infty) + \delta$
for all $t \geq T$.
\end{definition}

\subsection{Contract Specification}
\label{subsec:contract_specification}

We now specify contracts for each layer. For the high–level, we introduce two concrete
performance parameters: a maximum reference gap $\bar{r} \geq 0$ that bounds $\|r_k - r_{k-1}\|$, and a tolerance on the abstraction error $\varepsilon_E \geq 0$.

\begin{definition}
\label{def:contract_high}
The high-level contract
\begin{align}
\mathcal{C}_H=
(\bigwedge_k A_{\mathrm{mis}}^k,\,
\bigwedge_k (G_{\mathrm{ref}}^k \wedge G_{\mathrm{ISS}}^k))\label{eq:highlevelc}
\end{align}
is defined on trajectory set $\mathcal{B}_H$, where each per-step assumptions and guarantees are defined at time step $t_k = kT_s$ as:
\scriptsize{
\begin{align}
A_{\mathrm{mis}}^k := \Bigl\{\, (\hat{y}_k,&\, \hat{z}_k,\, y_k) 
  \in \mathbb{R}^{q} \times \mathbb{R}^{n-q} \times \mathbb{R}^{q}
  \;\Big|\; \nonumber \\ 
 \|\tilde{w}_k\| = &\|y_{k} 
  - \hat{f}(\hat{y}_{k-1|k-1},\, \hat{z}_{k-1|k-1},\, r_{k-1})\| 
  \leq \varepsilon_E \Bigr\}
  \label{eq:A_mis} \\
G_{\mathrm{ref}}^k &:= \Bigl\{\, r_k \in \mathbb{R}^{p}
  \;\Big|\; \|r_k - r_{k-1}\| \leq \bar{r} \Bigr\}
  \label{eq:G_ref} \\
G_{\mathrm{ISS}}^k := \Bigl\{\, y_k \in \mathbb{R}^{q}&
  \;\Big|\; \|y_k - y_{\mathrm{goal}}\|
  \leq \beta(\|y_0 - y_{\mathrm{goal}}\|,\, k)
  + \varepsilon_T(\varepsilon_E) \Bigr\}
  \label{eq:G_iss} 
\end{align}
}
\normalsize{for some $\beta \in \mathcal{KL}$ and
$\varepsilon_T \in \mathcal{K}$
(Definition~\ref{def:iss} for high level DT system with $(\mathbb{T},\, \mathcal{B},\, \xi,\, d)=(\mathbb{N},\, \mathcal{B}_H,\,
\yhat_{k|k} - y_{\mathrm{goal}},\, \tilde{w}_k)$). $y_{\mathrm{goal}} \in\mathbb{R}^{q}$ is the liveness goal in Definition~\ref{def:liveness_spec}. The high-level system has no
external disturbance assumption
($A_{\mathrm{ext}}^H = \top$). The reference
feasibility guarantee $G_{\mathrm{ref}}^k$ does not require
$A_{\mathrm{mis}}^k$.}
\end{definition}

For the low-level, we introduce two concrete
performance parameters: a \emph{settling time}
$\tau_{LL} > 0$, after which the tracker has
converged to a neighborhood of the reference, and a per-coordinate tracking tolerance $\varepsilon_{L,i}>0$, the radius of that
neighborhood in coordinate $i$ (collected as $\varepsilon_L=[\varepsilon_{L,i}]_i$). Both depend on the
controller design and the disturbance bound
$\mathcal{W}$.

\begin{definition}
\label{def:contract_low}
The low-level contract
\begin{align}
\mathcal{C}_L =
( \bigwedge_k (A_{\mathrm{env}}^k \wedge A_{\mathrm{ref}}^k),\,
 \bigwedge_k (G_{\mathrm{safe}}^k \wedge G_{\mathrm{track}}^k))\label{eq:lowlevelc}
\end{align}
is defined on trajectory set $\mathcal{B}_L$, where each per-step assumptions and guarantees are defined on $k$-th sampling interval by
$\mathbb{I}_k := [kT_s,\, (k{+}1)T_s]$ as:
\scriptsize{
\begin{align}
A_{\mathrm{env}}^{k} &:= \Bigl\{\, w : \mathbb{I}_k \to \mathbb{R}^d \;\Big|\; w(t) \in \mathcal{W} ,\;\;\forall\, t \in \mathbb{I}_k \Bigr\}, \label{eq:A_env} \\
A_{\mathrm{ref}}^{k} &:= G_{\mathrm{ref}}^k =  \Bigl\{\, r_k \in \mathbb{R}^{p}
  \;\Big|\; \|r_k - r_{k-1}\| \leq \bar{r} \Bigr\}
  \label{eq:A_ref} \\
G_{\mathrm{safe}}^{k} &:= \Bigl\{\, x : \mathbb{I}_k \to \mathbb{R}^n 
  \;\Big|\; x(t) \in \mathcal{X}_{\mathrm{safe}}, 
  \;\;\forall\, t \in \mathbb{I}_k \Bigr\}, \label{eq:G_safe} \\
G_{\mathrm{track}}^{k} &:= \Bigl\{\, (x,\, r_k) \,:\, 
  x : \mathbb{I}_k \to \mathbb{R}^n,\; r_k \in \mathbb{R}^p
  \;\Big|\; \nonumber \\ 
  &\big| h_{r,i}(x((k{+}1)T_s)) - r_{k,i}\big| \le \varepsilon_{L,i},\ \forall i \Bigr\},\label{eq:G_track}
\end{align}
}
\normalsize{under timing compatibility condition $C_{\mathrm{tss}}$ defined below.}
\end{definition}

\begin{definition}[Timing Compatibility]\label{def:timing_compatibility}
The timing compatibility condition is $C_{\mathrm{tss}}: T_s \ge \tau_{LL}$. 
\end{definition}

It ensures the high-level samples only after the low-level has
settled, and is a constraint on the shared sampling
architecture, that is, $
    \|h_r(x(t)) - r_k\| \leq \varepsilon_L
    \quad \forall\, t \in
    [kT_s + \tau_{LL},\, (k{+}1)T_s]$, which results in $G_{\mathrm{track}}^k$ in \eqref{eq:G_track}.

\subsection{Heterogeneous Well-Posed Interconnection}

For the local contracts $\mathcal{C}_H$ and
$\mathcal{C}_L$ to compose into the global
specification, the interconnection between the two
layers must be well-posed across the two
heterogeneous time domains.
This inter-layer conditions require a
condition beyond the well-posedness of the system-level
solution (which follows from local Lipschitzness of
the closed-loop dynamics): a vertical refinement condition under $C_{\mathrm{tss}}$ (cf. Figure \ref{fig:contract_framework}).

\begin{figure}[ht!]
    \centering
    \includegraphics[width=0.46\textwidth]{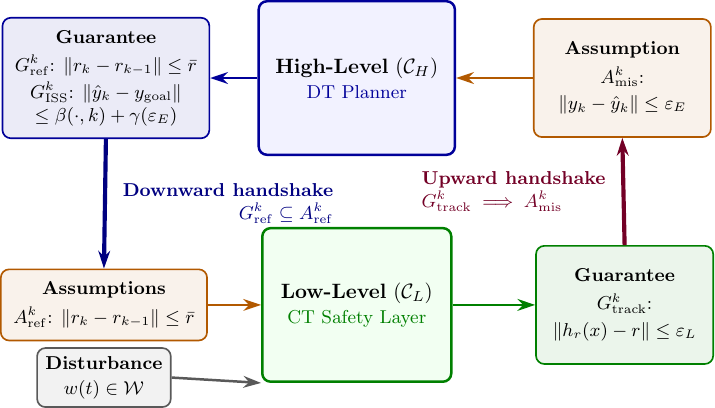}
    \caption{Vertical refinement for liveness specifications $\varphi_{\mathrm{live}}$: the ERG realizes the contract 
    hierarchy by ensuring that the low-level guarantees 
    $G_{\mathrm{safe}} \wedge G_{\mathrm{track}}$ under $A_{\mathrm{env}}$ 
    satisfy the high-level mismatch assumption $A_{\mathrm{mis}}$, 
    enabling the high-level system to guarantee $G_{\mathrm{live}}$.}
    \label{fig:contract_framework}
\end{figure}

\begin{definition}[Vertical Refinement]
\label{def:vertical_refinement}
The layered architecture satisfies the 
\emph{vertical refinement condition} $C_r$ for $\varphi_{\mathrm{safe}}$ and $\varphi_{\mathrm{live}}$ with tolerance 
$\varepsilon_H$ if the cross-domain handshakes 
are jointly feasible under $C_{\mathrm{tss}}$:
\begin{align}
 & \bigwedge_k (  G_{\mathrm{ref}}^k \implies  A_{\mathrm{ref}}^k )
    \label{eq:downward} \\
& \bigwedge_k (   G_{\mathrm{track}}^{k-1} 
    \implies A^k_{\mathrm{mis}}(\varepsilon_E) )
    \label{eq:upward}
\end{align}
for some $\varepsilon_E, \delta$ satisfying 
\begin{align}
    \varepsilon_T(\varepsilon_E)+\delta 
    < \varepsilon_H, \label{eq:verticalcomp}
\end{align}
where $\varepsilon_T$ is the 
ISS gain from $G_{\mathrm{ISS}}$.
\end{definition}

We define heterogeneous well-posedness for the layered interconnection.

\begin{definition}
\label{def:heterogeneous_wellposed}
The interconnection of contracts $\mathcal{C}_H$ and $\mathcal{C}_L$ for $\varphi_{\mathrm{safe}}$ and $\varphi_{\mathrm{live}}$ is 
\emph{recursively well-posed} by the layered system $\Sigma_H \triangleright\Sigma_L$ if it satisfies:
\begin{enumerate}[(i)]
\item \textbf{(initial conditions and environmental assumption)} $x(0),r(0)$, and $w(t)$ such that $A^0_\mathrm{ref}$ and $A_\mathrm{env}$ holds.
\item \textbf{(local contracts)} $\Sigma_H \models \mathcal{C}_H$, $\Sigma_L \models \mathcal{C}_L$ holds.
\item \textbf{(recursive feasibility)} The high-level layer remains feasible at every step. 
\item \textbf{(vertical refinement condition)} The time compatibility condition $C_{\mathrm{tss}}$ and the vertical refinement condition $C_r$ holds.
\end{enumerate}
\end{definition}

\begin{theorem}
\label{thm:layered_correctness}
Suppose the interconnection of contracts $\mathcal{C}_H$ and $\mathcal{C}_L$ is 
\emph{recursively well-posed} by layered system $\Sigma_H \triangleright\Sigma_L$ with tolerance 
$\varepsilon_H$ (Definition~\ref{def:heterogeneous_wellposed}).
Then Problem 1 is 
solved: the layered system satisfies 
$\varphi_{\mathrm{safe}}$ and 
$\varphi_{\mathrm{live}}$.
\end{theorem}

\begin{proof}
\emph{(a) Safety $\varphi_{\mathrm{safe}}$.} Given initial condition $x(0) \in~\mathcal{X}_{\mathrm{safe}}$ and $A_{\mathrm{env}}: w(t) \in \mathcal{W}$ for all $t$, the local contract satisfaction $\Sigma_L \models \mathcal{C}_L$ means that $G_{\mathrm{safe}}$ holds for any $r(t)$ for all $t$, $\varphi_{\mathrm{safe}} = G_{\mathrm{safe}}$ holds.

\emph{(b) Liveness $\varphi_{\mathrm{live}}$.} To prove global liveness, we need to guarantee $G_{\mathrm{ISS}}^{k_s}$ at each time step ${k_s}$. First, $G_{\mathrm{ref}}^k$ in \eqref{eq:G_ref} unconditionally holds for all $k$ without $A_{\mathrm{mis}}^k$ only from recursive feasibility condition (iii) (as $A^k_{\mathrm{mis}}$ is needed for $G^{k_s}_{\mathrm{ISS}}$ at past steps $k \le k_s$). Note that the recursive feasibility is not related to the contract $\mathcal{C}_H$. With the downward vertical refinement condition, the reference assumption $A_{\mathrm{ref}}^k$ holds for all time steps. Using the upward condition of $C_r$ for all time steps, $A_{\mathrm{mis}} = \bigwedge_k A_{\mathrm{mis}}^k$ holds. The local contract $\Sigma_H \models \mathcal{C}_H$ holds at each time step ${k_s}$ as, $\Sigma_H \models \mathcal{C}_H^{k_s}$: 
\begin{align}
\mathcal{C}_{H}^{k_{s}}
=\Bigl(\textstyle\bigwedge_{k\le k_{s}}A_{\mathrm{mis}}^{k},\ \
G_{\mathrm{ref}}^{k_{s}}\wedge G_{\mathrm{ISS}}^{k_{s}}\Bigr).
\label{eq:C_H}
\end{align}
since $G_{\mathrm{ISS}}^{k_{s}}$ depends on the mismatch being bounded at all
past steps $k\le k_{s}$. By the upward handshake,
$A_{\mathrm{mis}}=\bigwedge_{k}A_{\mathrm{mis}}^{k}$ holds for all $k$, so the
conjunction is available at every $k_{s}$ regardless of direction.

Setting $(\mathbb{T},\mathcal{B},\xi,d)=(\mathbb{N},\mathcal{B}_H,\hat y-y_{\mathrm{goal}},\tilde w_k)$
in Definition~\ref{def:iss} under $G^k_{\mathrm{ISS}}$, and using $\hat y_{k|k}=y_k$,
for any $\delta>0$ there exists finite $K(\delta)$ such that for all $k\ge K$,
\begin{align}
 & \|h_y(x(t_k))-y_{\mathrm{goal}}\|=\|y_k-y_{\mathrm{goal}}\| \nonumber \\
 & \le \beta\!\big(\|y_0-y_{\mathrm{goal}}\|,k\big)+\varepsilon_T(\varepsilon_E)
  \le \varepsilon_T(\varepsilon_E)+\delta .
  \label{eq:e2e}
\end{align}
With equation \eqref{eq:verticalcomp}, i.e.\ $\varepsilon_T(\varepsilon_E)+\delta<\varepsilon_H$,
there exists $K$ with $\|h_y(x(t_K))-y_{\mathrm{goal}}\|\le\varepsilon_H$,
i.e.\ $\exists\,T=KT_s$ such that $\varphi_{\mathrm{live}}$ holds. 
\end{proof}

\section{Layered Control Architecture}
\label{sec:lca}

This section implements the layered contract introduced in
Section~\ref{sec:heterogeneous_contracts} using a reference governor
combined with an MPC planner.  We first concretize each layer's
implementation and state the required assumptions, then verify that
the local conditions of
Definition~\ref{def:heterogeneous_wellposed} (and
Theorem~\ref{thm:layered_correctness}) are satisfied, ensuring that the closed-loop system meets the global contract with admissible
inputs $u(t) \in \mathcal{U}$.
Section~\ref{subsec:planning_layer} addresses the high-level
contract~\eqref{eq:highlevelc} and its recursive feasibility. Section \ref{sec:lca}, verifying \eqref{eq:C_H}, which implies $\Sigma_H \models \mathcal{C}_H$, through the MPC implementation. 

\subsection{Planning Layer: Discrete-Time System $\Sigma_H$}
\label{subsec:planning_layer}

Define the stage cost $\ell:\Rset^{q}\times\mathbb{R}\to\Rset_{\ge 0}$
and terminal cost $V_{f}:\Rset^{q}\to\Rset_{\ge 0}$, both continuous
and positive definite.  The \emph{MPC optimal control problem} is as follows. At time~$t_k$ with measurement
$y_k = h_y(x(t_k))$:
\begin{subequations}
\label{eq:mpc_ocp}
\begin{align}
    \min_{r_{k:k+N-1}} \quad
    & \sum_{j=0}^{N-1}
    \ell(\yhat_{k+j|k},\, r_{k+j})
    + V_f(\yhat_{k+N|k}) \\
    \text{s.t.} \quad
    & \eqref{eq:predicted_model} \text{ with } \yhat_{k|k} = y_k \\
    &\yhat_{k+j|k} \in \Yhat,\qquad j = 0,\dots,N\!-\!1,
    \label{eq:mpc_stage_constraints}\\
    & \yhat_{k+N|k} \in \Yf,
    \label{eq:mpc_terminal}\\
    & \|r_{k+j} - r_{k+j-1}\| \leq \bar{r}
    \label{eq:mpc_rate}
\end{align}
\end{subequations}
where
$\Yhat\subseteq\mathbb{R}^{q}$ is the (tightened) state constraint
set,
$\bar{r}$~is the (compact) reference constraint set,
and $\Yf\subseteq\Yhat$ is the terminal constraint set.%
\footnote{The terminal constraint in~\eqref{eq:mpc_terminal} uses
the terminal set $\Yf$, which is distinct from the $N$-step feasible
set $\YN$ defined in~\eqref{eq:YN_def} below.  In particular,
$\Yf\subseteq\YN$.}

Let $r^{*} = \bigl(r_{0}^{*},\dots,r_{N-1}^{*}\bigr)$
denote any minimising sequence and let $V_N^* : \YN \to \mathbb{R}_{\geq 0}$ be the
optimal value function, where the \emph{$N$-step feasible set} of
\eqref{eq:mpc_ocp} is
\begin{align}
    \YN := \bigl\{\,
    y \in \mathbb{R}^q :
    \text{\eqref{eq:mpc_ocp} is feasible with }
    \yhat_{k|k} = y
    \,\bigr\}.
    \label{eq:YN_def}
\end{align}
By Proposition~2.10(c) in~\cite{Rawlings2017-vw}, $\YN$ is closed.
Moreover, if the terminal set $\Yf$ is control invariant under~\eqref{eq:predicted_model}, then by
Proposition~2.10(b) in~\cite{Rawlings2017-vw} the feasible sets are
nested, $\YN \supseteq \mathcal{Y}_{N-1} \supseteq \cdots \supseteq \Yf$, and $\YN$ is positive invariant under the nominal MPC closed-loop $\yhat \mapsto r_0^*(\yhat)$.
Assumption~\ref{asm:mismatch} extends this to the disturbed setting.

\begin{assumption}[Robust Recursive Feasibility]
\label{asm:mismatch}
Let $A_{\mathrm{mis}} := \{\tilde w \in \mathbb{R}^q : \|\tilde w\| \le \eE\}$.
The terminal set $\Yf$ is robust control invariant for the predicted
dynamics~\eqref{eq:predicted_model} under $A_{\mathrm{mis}}$, and
the feasible set $\YN$ is robustly positive invariant in the following sense:
for every $\yhat\in\YN$, every $\zhat\in\mathcal{Z}$, every minimiser
$r^*=(r_0^*,\dots,r_{N-1}^*)$ of~\eqref{eq:mpc_ocp} at $\yhat_{k|k}=\yhat$,
and every $\tilde w\in A_{\mathrm{mis}}$,
\begin{align}
    \fhat(\yhat,\,\zhat,\,r_0^*) + \tilde w \;\in\; \YN.
    \label{eq:feasibility_margin}
\end{align}
\end{assumption}
Concretely, Assumption~\ref{asm:mismatch} guarantees that whenever
\eqref{eq:mpc_ocp} is feasible at time~$k$, it remains feasible at
time~$k+1$ despite any bounded model--reality mismatch $\tilde w \in
A_{\mathrm{mis}}$.\footnote{Robust recursive feasibility
can be enforced by constraint tightening (e.g.,
\cite{ChisciEtAl2001,MayneSeronRakovic2005}). It is not required for safety, which the low-level layer maintains unilaterally. If
feasibility is lost at step~$k$, the fallback policy $r_k = r_{k-1}$
keeps the system safe, though liveness is suspended until a feasible solution is recovered.}

\begin{assumption}[Planner Stability]
\label{asm:mpc_stability}
The MPC~\eqref{eq:mpc_ocp} satisfies the following, uniformly in the exogenous data
$(\zhat_k, r_k) \in \mathcal{Z} \times \mathcal{R}$:
\begin{enumerate}[(i)]
    \item There exist
    $\alpha_1, \alpha_2 \in \mathcal{K}_\infty$
    such that, for all $\yhat \in \YN$,
    \begin{align}
        \alpha_1\bigl(\|\yhat - \ygoal\|\bigr)
        \;\leq\; V_N^*(\yhat)
        \;\leq\;
        \alpha_2\bigl(\|\yhat - \ygoal\|\bigr)
        \label{eq:sandwich_bounds}
    \end{align}

\item There exists $\alpha_{3}\in\mathcal{K}_{\infty}$ such that, for all
$\hat{y}_{k}\in\mathcal{Y}_{N}$, the MPC feedback $r_{k}=r_{0}^{*}(\hat{y}_{k})$
and the associated $\hat{z}_{k}$ satisfy
\begin{equation}\label{eq:descent_condition}
V_{N}^{*}\!\bigl(\hat{f}(\hat{y}_{k},\hat{z}_{k},r_{0}^{*}(\hat{y}_{k}))\bigr)
-V_{N}^{*}(\hat{y}_{k})
\le -\alpha_{3}\bigl(\|\hat{y}_{k}-y_{\mathrm{goal}}\|\bigr).
\end{equation}

    \item $V_N^*$ is
    Lipschitz on $\YN$ with constant $L_V > 0$:
    \begin{align}
        \bigl|V_N^*(\yhat)
        - V_N^*(\yhat')\bigr|
        \;\leq\; L_V\,\|\yhat - \yhat'\|
        \label{eq:lipschitz}
    \end{align}
\end{enumerate}
\end{assumption}

\begin{remark}\label{asm:mpcstandard}
Assumption~\ref{asm:mpc_stability} is satisfied
under standard MPC terminal conditions
(see Chapter 2 of \cite{Rawlings2017-vw}). 
\end{remark}

\begin{remark}
\label{rem:unconstrained}
When $\Yhat = \Yf = \mathbb{R}^q$, the feasible
set is $\YN = \mathbb{R}^q$ for any horizon~$N$,
and Assumption~\ref{asm:mismatch} is satisfied
trivially. The low-level controller then bears
sole responsibility for keeping the physical
state in the region where
Assumption~\ref{asm:mpc_stability} holds.
\end{remark}

The one-step error bound in \eqref{eq:descent_condition} can be extended to the case with disturbance $\tilde{w}_k$ from the mismatch between $f$ and $\fhat$ using \eqref{eq:lipschitz}.

\begin{lemma}
\label{lem:dissipation}
Under
Assumptions~\ref{asm:mismatch}
and~\ref{asm:mpc_stability}, set
$\hat{y}_{k} := \hat{f}(\hat{y}_{k-1|k-1}, \hat{z}_{k-1|k-1}, r_{k-1}) + \tilde{w}_{k}$
with $A_{\mathrm{mis}}: \|\tilde{w}_k\| \le \varepsilon_E$. Then for all
$\yhat_k \in \YN$:
\begin{align}
  V_N^*(\hat y_{k|k}) &- V_N^*(\hat y_{k-1|k-1})\nonumber \\
 & \le -\alpha_3\big(\|\hat y_{k-1|k-1}-y_{\mathrm{goal}}\|\big) + L_V\varepsilon_E.
  \label{eq:dissipation}
\end{align}
\end{lemma}
Using Lemma \ref{lem:dissipation}, the following proposition provides the conditions under which $G_{\mathrm{ISS}}^k$ holds.

\begin{proposition}[ISS of the Planning Loop]
\label{prop:mpc_iss}
Under
Assumptions~\ref{asm:mismatch}
and~\ref{asm:mpc_stability}, set
$\hat{y}_{k} := \hat{f}(\hat{y}_{k-1|k-1}, \hat{z}_{k-1|k-1}, r_{k-1}) + \tilde{w}_{k}$
with $A_{\mathrm{mis}}: \|\tilde{w}_k\| \le \varepsilon_E$. Then the state
$\hat{y}_{k|k}$ satisfies $G_{\mathrm{ISS}}^{k}$ in \eqref{eq:G_iss} for some $\beta \in \mathcal{KL}$, where
    $\eT = \alpha_1^{-1}(\alpha_2(\gamma(\eE)))$
    with $\gamma=\; \alpha_3^{-1}(L_V\, s)$.
\end{proposition}
\begin{proof}
With the bounds~\eqref{eq:sandwich_bounds}, \eqref{eq:descent_condition},
and \eqref{eq:dissipation} from Lemma~\ref{lem:dissipation}, $V_N^*$
is an ISS-Lyapunov function with disturbance input $L_V\,\eE$.
The standard ISS-Lyapunov theorem~\cite{JiangWang2001} gives
\begin{align}
    V_N^*(\yhat_k)
    \;\leq\;
    \tilde{\beta}(V_N^*(\yhat_0),\, k)
    + \alpha_2(\gamma(\eE))
\end{align}
for some $\tilde{\beta} \in \mathcal{KL}$ and the $\mathcal{K}$-class
ISS gain
\begin{align}
    \gamma(s) \;:=\; \alpha_3^{-1}(L_V\, s).
    \label{eq:iss_gain}
\end{align}
Applying the sandwich bounds~\eqref{eq:sandwich_bounds} to both
sides yields
\begin{align}
    \|\yhat_k - \ygoal\|
    \;&\leq\;
    \alpha_1^{-1}\!\bigl(
    \tilde{\beta}(\alpha_2(\|\yhat_0 - \ygoal\|),\, k)
    + \alpha_2(\gamma(\eE))\bigr) \nonumber \\
    &\leq\;
    \beta(\|\yhat_0 - \ygoal\|,\, k)
    + \eT,
    \label{eq:eH_def}
\end{align}
where $\beta \in \mathcal{KL}$ absorbs
$\alpha_1^{-1} \circ \tilde{\beta} \circ
\alpha_2$ and $\eT$:
\begin{align}
    \eT
    \;=\;
    \alpha_1^{-1}\!\bigl(
    \alpha_2\bigl(
    \alpha_3^{-1}(L_V\,\eE)
    \bigr)\bigr).
    \label{eq:eH_explicit}
\end{align}
\end{proof}
Under $A_{\mathrm{mis}}^{k} ,\forall k \geq k_s$,
Proposition~\ref{prop:mpc_iss} delivers
$G_{\mathrm{ISS}}^{k_s}$.
The reference guarantee $G_{\mathrm{ref}}^{k_s}$
holds by~\eqref{eq:mpc_rate}. 

\begin{corollary}[High-Level Contract]
\label{cor:high_level_contract}
Under Assumptions \ref{asm:mismatch}--\ref{asm:mpc_stability} on $\Sigma_H$, $\Sigma_H \models \mathcal{C}_H$ holds.
\end{corollary}

\subsection{Safety Layer: Continuous-Time System $\Sigma_L$}
\label{subsec:safety_layer}

The low-level safety layer $\Sigma_L$ comprises a tracking controller and an
explicit reference governor (ERG). The tracker generates the control input
$u(t) = \kappa(x(t), v(t))$ (cf.~\eqref{eq:tracking_controller}) to regulate the plant state
toward a filtered reference $v(t)$ produced by the ERG; together they must
satisfy $\mathcal{C}_L$. We develop the layer directly in the linear--quadratic
setting that fits the HESS application of Section~\ref{sec:HESS}; the general
nonlinear formulation, of which this is the closed-form specialization, is given
in ~\appref{app:nonlinear}.

\subsubsection{Inner Loop}
\label{sec:inner_loop}
The tracker regulates the tracked element $h_r(x(t)) \in \mathbb{R}^p$ to follow
the ERG-filtered reference $v(t) \in \mathbb{R}^p$, with tracking error
$e(t) := h_r(x(t)) - v(t)$. Provided the controller $\kappa$ admits an
error-coordinate representation in which the disturbance enters additively, the
closed-loop error dynamics take the linear form
\begin{equation}
\dot e = A e + B w - B_v \psi_{\mathrm{dyn}}(\dot v, \ddot v, \ldots),
\label{eq:err_dyn_lin}
\end{equation}
where $A$ captures the autonomous error evolution under the tracking controller,
$B$ the disturbance channel, and $\psi_{\mathrm{dyn}}$ the coupling from the
ERG-induced reference motion. The \emph{frozen} dynamics, obtained with
$\dot v = 0$ (reference held constant over a sampling interval), are
\begin{equation}
\dot e = A e + B w, \quad \|e^{At}\| \le m\,e^{-\lambda_e t},
\label{eq:frozen_lin}
\end{equation}
with $A$ Hurwitz and $m \ge 1,\ \lambda_e > 0$. The disturbance bound $\|w\| \le W_{\max}$ yields
$\|B w\| \le H_{\max} := \|B\|\,W_{\max}$.

\begin{assumption}[Quadratic ISS]
\label{asm:quad_iss}
There exists $P \succ 0$ with $A^\top P + P A \prec 0$ such that
$V(e) = e^\top P e$ satisfies, along the frozen dynamics~\eqref{eq:frozen_lin},
\begin{align}
&\lambda_{\min}(P)\|e\|^2 \le V(e) \le \lambda_{\max}(P)\|e\|^2,\\
&\dot V(e)\le -\alpha_d(\|e\|)+\sigma(\|Bw\|),
\label{eq:quad_iss}
\end{align}
and $\|Bw\|\le H_{\max}:=\|B\|\,W_{\max}$ for some $\alpha_d, \sigma \in \mathcal{K}_\infty$.
\end{assumption}

Under Assumption~\ref{asm:quad_iss}, the ISS ultimate sublevel threshold has the
closed form
\begin{equation}
\bar V_h(H_{\max}) = \lambda_{\max}(P)\,\bigl(\alpha_d^{-1}(\sigma(H_{\max}))\bigr)^2,
\label{eq:Vbar_def}
\end{equation}
and the sublevel set $\Omega_h := \{e : V(e) \le \bar V_h(H_{\max})\}$ is forward
invariant, with every trajectory entering it in finite time. Since $A$ is
Hurwitz, the ISS bound~\eqref{eq:frozen_lin} gives the envelope in Definition~\ref{def:iss} with $\beta(s,t) = m\,e^{-\lambda_e t} s$ and
$\gamma(s) = \gamma_{\mathrm{ISS}} s$ with $\gamma_{\mathrm{ISS}}:={m}/\lambda_e$:
\begin{equation}
\|e(t)\|\le m\,e^{-\lambda_e t}\|e(0)\|+\gamma_{\mathrm{ISS}}\,H_{\max}.
\label{eq:iss_envelope}
\end{equation}

\subsubsection{Explicit Reference Governor}
\label{sec:erg}
We introduce the ERG~\cite{GaroneNicotraERG, NicotraUAVERG} to enforce the safety
constraints~\eqref{eq:safe_set_compact} on the physical plant. Writing the $i$-th row of
$Cx \le d$ in error coordinates ($x = e + v$, with $v$ zero-padded onto the
untracked coordinates) as $c_{e,i}^\top e \le d_i(v)$, $d_i(v) := d_i - c_i^\top v$, the Lyapunov threshold in~\cite{NicotraUAVERG} (cf. general form in \appref{app:subdefthre}) is 
\begin{equation}
\Gamma_i(v) = \frac{d_i(v)^2}{c_{e,i}^\top P^{-1} c_{e,i}},
\qquad \Gamma(v) := \min_i \Gamma_i(v).
\label{eq:Gamma_i}
\end{equation}
Geometrically, $\Gamma_i(v)$ is the largest $V$-sublevel set contained in the
$i$-th constraint half-space. The commanded reference $r(t)$ from the ZOH is
filtered to $v(t)$ by the ERG dynamics
\begin{equation}
\dot v = \kappa_{\mathrm{erg}}\cdot \Delta(e,v)\cdot \rho(r,v),
\quad \Delta(e,v) := \max\bigl(0,\ \Gamma(v) - V(e)\bigr),
\label{eq:erg_dynamics}
\end{equation}
with gain $\kappa_{\mathrm{erg}} > 0$ and navigation field
$\rho(r,v) = \rho_{\mathrm{att}}(r,v) + \sum_i \rho_{\mathrm{rep},i}(v)$
combining attraction toward $r$ and repulsion from constraint boundaries. The attraction field steers $v$ toward the target $r$,
\begin{equation}
\rho_{\mathrm{att}}(r,v) =
\begin{cases}
\dfrac{r-v}{\|r-v\|} & \|r-v\| \ge \eta,\\[6pt]
\dfrac{r-v}{\eta} & \|r-v\| < \eta,
\end{cases}
\label{eq:rho_att}
\end{equation}
where $\eta$ is the smoothing radius (Assumption~\ref{asm:smoothing}); for each
constraint boundary $i$ the repulsive field is
\begin{equation}
\rho_{\mathrm{rep},i}(v) = +\,\eta_i\,
   \frac{\nabla_v \Gamma_i(v)}{\|\nabla_v \Gamma_i(v)\|},
\label{eq:rho-rep}
\end{equation}
with $\eta_i$ the repulsion strength, so that the combined field is
\begin{equation}
\rho(r,v) = \rho_{\mathrm{att}}(r,v) + \sum_i \rho_{\mathrm{rep},i}(v).
\label{eq:rho_combined}
\end{equation}
The design condition
$\delta_{\mathrm{rep}} := \sum_i \eta_i < 1$ ensures forward progress. When
$\Delta(e,v) = 0$ the reference is frozen ($\dot v = 0$) and $V(e)$ decreases by
Assumption~\ref{asm:quad_iss} until $\Delta(e,v) > 0$, giving robust constraint
satisfaction regardless of $r$. The feasible reference set is
$\mathcal{V} := \{v \in \mathbb{R}^p : \Gamma(v) > 0\}$.

\begin{assumption}[Admissible Disturbance]
\label{asm:adm_dist}
The disturbance bound $H_{\max}$ is admissible:
$\bar V_h(H_{\max}) < \inf_{v \in \mathcal{V}} \Gamma(v)$, so that
$\Omega_h \times \mathcal{V} \subset \tilde K := \{(e,v) : V(e) \le \Gamma(v)\}$.
\end{assumption}

Figure~\ref{fig:rci_sets} illustrates $\tilde K$ in the augmented
state space.

\begin{theorem}[Robust Invariance of ERG]
\label{thm:erg_invariance}
Under Assumptions~\ref{asm:quad_iss}--\ref{asm:adm_dist}, the set
$\tilde K = \{(e,v) : e^\top P e \le \Gamma(v)\}$ is robustly forward invariant
for the composite dynamics~\eqref{eq:err_dyn_lin},~\eqref{eq:erg_dynamics} for
all $\|w\| \le W_{\max}$ and any commanded reference $r(\cdot)$.
\end{theorem}

\begin{proof}
On the boundary $\partial\tilde K$, where $V(e) = \Gamma(v)$, we have
$\Delta(e,v) = 0$, so the ERG law~\eqref{eq:erg_dynamics} gives $\dot v = 0$ and
the error obeys the frozen dynamics~\eqref{eq:frozen_lin}. Since
$V(e) = \Gamma(v) \ge \inf_{v\in\mathcal V}\Gamma(v) > \bar V_h$ by
Assumption~\ref{asm:adm_dist}, the state lies outside $\Omega_h$,
$\alpha_d(\|e\|) > \sigma(H_{\max})$ and $\dot V(e)<0$ by~\eqref{eq:quad_iss}.
With $\dot v = 0$ the threshold $\Gamma(v)$ is constant, so
$\tfrac{d}{dt}\bigl(V(e) - \Gamma(v)\bigr) = \dot V(e) < 0$ on
$\partial\tilde K$. Hence $\tilde K$ is robustly forward invariant.
\end{proof}

\begin{corollary}[Low-Level Contract]
\label{cor:lowlevel}
Under Assumptions~\ref{asm:quad_iss}--\ref{asm:adm_dist}, $C_{\mathrm{tss}}$, and
$(e(kT_s), v(kT_s)) \in \tilde K$ for all $k$, the low-level contract
$\Sigma_L \models \mathcal{C}_L$ holds:
\begin{enumerate}
\item[(a)] \emph{(Safety)} $G_{\mathrm{safe}}$ holds for any reference $r(t)$:
$x(t) \in \mathcal{X}_{\mathrm{safe}}$ and $u(t) = \kappa(x(t),v(t)) \in \mathcal{U}$
for all $t \in I_k$ (requires $A_{\mathrm{env}}$, not $A_{\mathrm{ref}}$).
\item[(b)] \emph{(Tracking)} Under $A_{\mathrm{ref}}$ and $C_{\mathrm{tss}}$,
$G^k_{\mathrm{track}}$ holds for all $k$ with the per-coordinate bound
\begin{equation}
\varepsilon_{L,i} := \sqrt{\bar V_h(H_{\max})\,[P^{-1}]_{ii}}.
\label{eq:eps_L}
\end{equation}
\end{enumerate}
\end{corollary}

\begin{proof}
(a) By Theorem~\ref{thm:erg_invariance}, $\tilde K$ is robustly forward
invariant, so $V(e) \le \Gamma(v) \le \Gamma_i(v)$ for every $i$, giving
$c_{e,i}^\top e \le d_i(v)$ and hence $x(t)\in\mathcal{X}_{\mathrm{safe}}$,
$u(t)\in\mathcal{U}$. (b) Once $v(t) = r$ for $t \ge t_0$, the frozen
dynamics~\eqref{eq:frozen_lin} drive the error into $\Omega_h$ within time
$\tau_{LL}$; the Cauchy--Schwarz projection of the ellipsoid
$\{e : e^\top P e \le \bar V_h\}$ onto coordinate $i$ gives
$|e_i| \le \sqrt{\bar V_h\,[P^{-1}]_{ii}}$, i.e.~\eqref{eq:eps_L}.
\end{proof}

\begin{figure}[ht]
    \centering
    \includegraphics[width=0.9\linewidth]{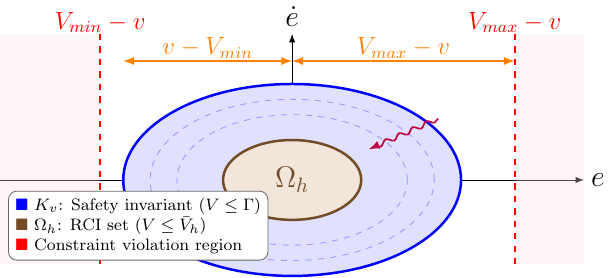}
	\caption{Invariance in the augmented $(e,v)$-space. The set
	$\tilde K = \{(e,v) : V(e) \le \Gamma(v)\}$ is safe; for fixed $v$ the slice
	$\mathcal{K}_v = \{e : e^\top P e \le \Gamma(v)\}$ (blue) shrinks as $v$
	approaches constraint boundaries where $\Gamma(v)$ decreases. The ISS
	ultimate bound $\Omega_h = \{e : e^\top P e \le \bar V_h\}$ (brown) satisfies
	$\Omega_h \subset \mathcal{K}_v$ for all admissible $v$
	(Assumption~\ref{asm:adm_dist}). On $\partial\tilde K$ the ERG enforces
	$\dot v = 0$, and the frozen quadratic ISS dynamics yield $\dot V(e) < 0$,
	ensuring invariance.}
    \label{fig:rci_sets}
\end{figure}


\subsubsection{Settling Time and Timing Condition}
\label{subsubsec:erg}
The tracking guarantee $G^k_{\mathrm{track}}$ requires timing compatibility
$C_{\mathrm{tss}} : T_s \ge \tau_{LL}$ (Definition~\ref{def:timing_compatibility}). To obtain an
explicit estimate of $\tau_{LL}$ we impose three operating conditions on the
ERG.

\begin{assumption}[Negligible Smoothing]
\label{asm:smoothing}
The smoothing radius $\eta$ in~\eqref{eq:rho_att} satisfies $\eta \ll \bar r$;
the reference is treated as arrived once $\|v - r\| < \eta$, so $\dot v \approx 0$
and the frozen dynamics~\eqref{eq:frozen_lin} govern the error.
\end{assumption}

\begin{assumption}[Reference Admissibility]
\label{asm:ref_adm}
Each target lies in the interior of the feasible set: there is $\delta_r > 0$
with $r_k \in \mathcal{V}_{\delta_r} := \{v\in\mathcal V : \mathrm{dist}(v,\partial\mathcal V)\ge\delta_r\}$
for all $k$. This is a liveness-only requirement; safety is unaffected.
\end{assumption}

\begin{assumption}[Guaranteed Progress]
\label{asm:progress}
There exist a uniform margin $\underline r > 0$ and $\tau_1 > 0$ such that
$\Gamma(v(t)) - V(e(t)) > \underline r$ for all $t \in [t_k, t_k + \tau_1]$,
implying $\underline\kappa\,\underline r \le \|\dot v\| \le \bar\kappa\,\bar\Gamma$
with $\underline\kappa = \kappa_{\mathrm{erg}}(1-\delta_{\mathrm{rep}})$,
$\bar\kappa = \kappa_{\mathrm{erg}}(1+\delta_{\mathrm{rep}})$, and
$\bar\Gamma := \sup_{v\in\mathcal V}\Gamma(v)$ the uniform threshold bound over
the feasible set.
\end{assumption}

\begin{proposition}[Settling Time]
\label{prop:settling}
Under Assumptions~\ref{asm:quad_iss}--\ref{asm:progress}, for any reference step
$\|r_{k+1} - r_k\| \le \bar{r}$ and any tolerance $\delta > 0$, the error is
bounded by $\|e(t)\| \le (1+\delta)\varepsilon = (1+\delta)\,\gamma_{\mathrm{ISS}} H_{\max}$
for all $t \ge t_k + \tau_{LL}$, where
\begin{equation}
\tau_{LL} :=
    \underbrace{\frac{\bar{r} + \varepsilon}{\underline{\kappa}\,\underline{r}}}_{\tau_1\ (\text{transit})}
    + \underbrace{\frac{1}{\lambda_e}\ln\!\Bigl(\frac{m\, e_{\mathrm{peak}}}{\delta\,\varepsilon}\Bigr)}_{\tau_2\ (\text{decay})},
\label{eq:tau_LL}
\end{equation}
with peak error $e_{\mathrm{peak}} := m\,e^{-\lambda_e \tau_1}(\bar{r}+\varepsilon)
  + \gamma_{\mathrm{ISS}}(H_{\max}+M)$ and $M := c_\psi \bar{\kappa}\bar{\Gamma}$,
where $c_\psi$ bounds the reference-coupling channel
$\|B_v \psi_{\mathrm{dyn}}\|$ per unit threshold rate.
\end{proposition}

\begin{proof}
The error envelope is bounded across the two sequential ERG phases---transit ($v\rightarrow r_k$) and decay ($\dot v = 0$), illustrated in Fig.~\ref{fig:tracking_decomp}; the phase-wise envelope bounds yielding \eqref{eq:tau_LL} are derived in \appref{app:settling}.
\end{proof}
Combining Corollary~\ref{cor:lowlevel} (low-level contract) with the timing
compatibility condition $C_{\mathrm{tss}}:T_s\ge\tau_{LL}$ (Prop.~\ref{prop:settling})
and the vertical refinement conditions of Definition~\ref{def:vertical_refinement},
the interconnection is recursively well-posed, so Theorem~\ref{thm:layered_correctness} yields
$\varphi_{\mathrm{safe}}\wedge\varphi_{\mathrm{live}}$ for $\Sigma_H\triangleright\Sigma_L$.

\begin{figure}[t!]
    \centering
    \includegraphics[width=1\linewidth]{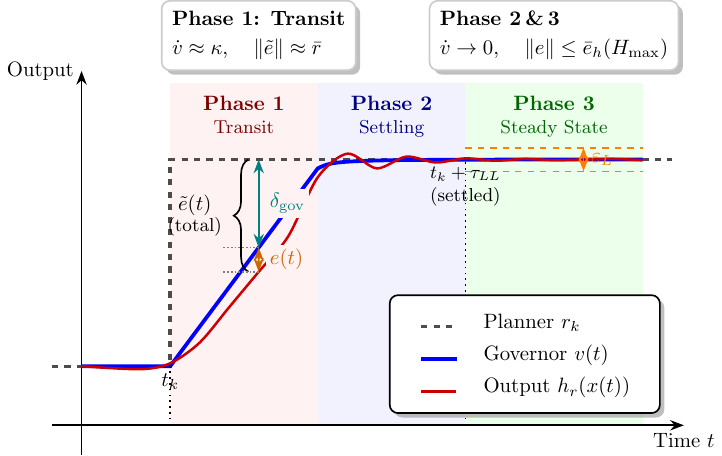}
\caption{\textbf{Hierarchical decomposition of tracking errors.} 
The Total Planner Deviation aggregates the Governor Lag $\|v(t) - r_k\|$ and the Inner-Loop Control Error $e(t)$.
The state $V_{\mathrm{gr}}(t)$ may transiently exceed the tracking tube $\pm\varepsilon$ but settles before the next sample.}
    \label{fig:tracking_decomp}
\end{figure}


\section{Case Study: Hybrid Energy Storage System}
\label{sec:HESS}

We instantiate the heterogeneous contract framework on a Hybrid Energy 
Storage System (HESS) comprising a battery and a supercapacitor connected to a DC bus,
which exhibit a natural slow--fast split: the battery supplies sustained energy
while the supercapacitor handles fast transients. The HESS is an instance of the
continuous-time plant with state $x := [V_{\mathrm{gr}},\, I_S,\, I_B,$ $\, E_S,\, E_B]^\top
\in \mathcal{X} \subseteq \mathbb{R}^5$,
control input $u := [u_S,\, u_B]^\top \in \mathcal{U}
\subseteq \mathbb{R}^2$,
disturbance $w(t) \in \mathcal{W} \subseteq \mathbb{R}$,
and known load $d(t) \in \mathbb{R}$, treated as an explicit function of $t$ rather than an input or disturbance since its profile is known a priori. The resulting dynamics, with substituted time-dependent term $d$, are assumed to satisfy the continuity condition in Definition~\ref{def:ct_plant}.
The dynamics $\dot{x} = f(x,\, u,\, w)$
in~\eqref{eq:ct_system} take the form:
\begin{subequations}
\label{eq:HESS_dynamics}
\begin{align}
    \dot{V}_{\mathrm{gr}}
        &= \frac{1}{C_{\mathrm{bus}}}
           \bigl(I_S + I_B + d\bigr),
        \label{eq:HESS_V} \\
    \dot{I}_S &= u_S + w(t),
        \label{eq:HESS_IS} \\
    \dot{I}_B &= u_B,
        \label{eq:HESS_IB} \\
    \dot{E}_S &= c_S\, V_{\mathrm{gr}}\, I_S,
        \label{eq:HESS_ES} \\
    \dot{E}_B &= c_B\, V_{\mathrm{gr}}\, I_B,
        \label{eq:HESS_EB}
\end{align}
\end{subequations}
where $V_{\mathrm{gr}}$ is the DC bus voltage,
$C_{\mathrm{bus}}$ the bus capacitance,
$I_S$ and $I_B$ the supercapacitor and battery currents,
$E_S$ and $E_B$ their states of charge, and
$c_S, c_B > 0$ the voltage-to-energy
conversion coefficients.
The load $d(t)$ is known, continuously differentiable,
and bounded ($|d(t)| \leq \rho_d$);
the MPC incorporates forecasts $\hat{d}[k]$,
while forecast errors contribute to the one-step
prediction error $\tilde{w}_k$
in~\eqref{eq:realised_dynamics}.
The disturbance satisfies $w(t) \in \mathcal{W} := [-W_{\max}, W_{\max}]$
for all $t \ge 0$, corresponding to the environmental assumption $A_{\mathrm{env}}$~\eqref{eq:A_env}. The dynamics exhibit a three-tier slow--fast cascade: the fast currents $I_S, I_B$
drive the intermediate bus voltage $V_{\mathrm{gr}}$ through~\eqref{eq:HESS_V},
which in turn feeds the slow energy states via the conversion couplings
$\dot E_S = c_S V_{\mathrm{gr}} I_S$ and $\dot E_B = c_B V_{\mathrm{gr}} I_B$,
with the load $d(t)$ entering at the voltage node.

\textbf{Output maps.}
The tracked output
$h_r(x) = P_r\, x \in \mathbb{R}^2$
in Definition~\ref{def:ct_plant} selects the
fast states that the low-level controller
$u = \kappa(x,\, v)$ as in~\eqref{eq:tracking_controller}
regulates the fast states (the currents, whose dynamics are not modeled in the high-level), i.e.,
\begin{align}
    h_r(x)
    = P_r\, x
    = [V_{\mathrm{gr}},\, I_B]^\top,
    \qquad
    P_r \in \mathbb{R}^{2 \times 5}.
    \label{eq:HESS_hr}
\end{align}
The planner measurement
$y_k = h_y(x(t_k)) \in \mathbb{R}^2$
in Definition~\ref{def:dt_planner} extracts
the slow observable states:
\begin{align}
    y_k
    = h_y(x(t_k))
    = [E_B(t_k)\, E_S(t_k)]^\top.
    \label{eq:HESS_hy}
\end{align}
\textbf{Specifications.}
As in Problem 1, the global requirement decomposes as
$\varphi_{\mathrm{global}}
= \varphi_{\mathrm{safe}} \land \varphi_{\mathrm{live}}$, where the safety specifications are imposed as physical limits on currents and the bus voltage (hard constraints):
\begin{align}
   \varphi_{\mathrm{safe}}:\quad
    &\forall t \geq 0: (V_{\mathrm{gr}},\, I_S,\, I_B)(t)
    \in \mathcal{X}_{\mathrm{safe}},
    \label{eq:HESS_safety_spec}\\
    &\mathcal{X}_{\mathrm{safe}} := \mathcal{R}_V \times \mathcal{R}_S
       \times \mathcal{R}_B,
    \label{eq:HESS_safe_set}
\end{align}
with $\mathcal{R}_V :=[V_{\min},V_{\max}]=[V_{\mathrm{nom}}-\Delta V,\;V_{\mathrm{nom}}+\Delta V]$,
$\mathcal{R}_S := [-\bar{I}_S,\, \bar{I}_S]$, and
$\mathcal{R}_B := [-\bar{I}_B,\, \bar{I}_B]$. Moreover,
\begin{align}
    \varphi_{\mathrm{live}}:\quad
    \exists\, T:\;
    |E_B(t) - E_B^{\mathrm{goal}}|
    \leq \varepsilon_H
    \quad \forall\, t \geq T,
    \label{eq:HESS_liveness_spec}
\end{align}
where the battery's SOC must achieve its goal in a finite time. This is the instance of Definition~\ref{def:liveness_spec} obtained with the seminorm $\|\hat y\|_{\mathrm{live}}=|E_B-E_B^{\mathrm{goal}}|$ selecting the regulated (battery) coordinate; the supercapacitor energy carries no target.
We also include input constraints.
\begin{align}
    \mathcal{U}
    := \bigl\{(u_S,\, u_B)
       : |u_S| \leq U_{S,\max},\;
         |u_B| \leq \bar{U}_B\bigr\}.
    \label{eq:HESS_input_constraints}
\end{align}
In addition to the hard constraints on current and voltage levels above, we introduce a soft constraint on the states of charge (SOC) below. This soft constraint is imposed only on the abstracted model; whether it should be regarded as part of the specification is therefore context-dependent.

\emph{Soft constraints (SOC bounds)}\;
$\varphi_{\mathrm{safe}}^{\mathrm{soft}}$:
\begin{align}
    \varphi_{\mathrm{safe}}^{\mathrm{soft}}:\quad
    (E_S,\, E_B)(t)
    \in [\underline{E}_S,\, \bar{E}_S]
       \times [\underline{E}_B,\, \bar{E}_B].
    \label{eq:HESS_soft_safety}
\end{align}

\begin{remark}\label{rem:superbatt}
The supercapacitor has no explicit target; it acts as a safety-layer buffer absorbing
high-frequency disturbances the slower battery cannot track. This is a component-level
safety--liveness decomposition: supercapacitor for safety (disturbance rejection),
battery for liveness (energy target).
\end{remark}

\textbf{Abstract model.} For regulated voltage operation,
the high-level system model assumes that $V_{\mathrm{gr}}^{\mathrm{ref}}[k] = V_{\mathrm{nom}}$
is constant; only $I_B^{\mathrm{ref}}[k]$ varies. This and equation~\eqref{eq:HESS_V} gives us an equation $I_S[k] = - I_B[k] - \hat{d}[k]$. Using this equation to \eqref{eq:HESS_EB}--\eqref{eq:HESS_ES}, the following is the dynamics of the planner that operates on slow states $\hat{y}_k := [E_B[k], E_S[k]]^\top$, 
abstracting fast voltage/current dynamics
\begin{align}
    \hat{y}_{k+1} &= \fhat(\yhat_{k|k},\,
                  \zhat_{k|k},\, r_{k}) \nonumber\\
                  &:= 
    \begin{bmatrix}
        E_B[k] + T_s c_B V_{\mathrm{nom}} I_B[k] \\
        E_S[k] + T_s c_S V_{\mathrm{nom}} (-I_B[k] - \hat{d}[k])
    \end{bmatrix} .\label{eq:bmatrixf}
\end{align} 
The prediction model uses the sampled disturbance $\hat d_k = d(t_k)$
held over the interval, with the inter-sample variation $\tilde d(t) = d(t) -
d(t_k)$ bounded by the load's Lipschitz constant and subsumed into $W_{\max}$.

With these settings, the local contracts follow from Corollary~\ref{cor:high_level_contract}
(high level) and Corollary~\ref{cor:lowlevel} (low-level tracker), and
Theorem~\ref{thm:layered_correctness} gives $\varphi_{\mathrm{safe}}\wedge\varphi_{\mathrm{live}}$;
the input constraints~\eqref{eq:HESS_input_constraints} are verified in
Proposition~\ref{cor:low}. The SOC bounds are operational—a violation signals an
energy-management failure, not immediate damage—so they are treated as soft: the
low-level controller~\eqref{eq:tracking_controller} does not enforce them, while the
high-level system folds them into reference generation to steer back within bounds
over time.

\subsection{High-Level Subsystem $\Sigma_H$.}
\label{subsec:HESS_high_level}

The MPC problem \eqref{eq:mpc_ocp} is instantiated as
\begin{subequations}
\label{eq:HESS_MPC}
\begin{align}
    \min_{r_{k:k+N-1}} \quad & \sum_{j=0}^{N-1} \|E_B[k+j+1] - E_B^{\mathrm{goal}}\|_Q^2 \\
\text{s.t.}\quad &\hat{y}_{k+j+1} = \hat{f}(\hat{y}_{k+j}, \hat{z}_{k+j}, r_{k+j}),
\\
& r_{k+j} = I_B^{\mathrm{ref}}[k+j],\quad I_B[k+j] \in \mathcal{R}_B,\label{eq:battbound} \\
    &\hspace{-1.5cm}|I_B[k+j] - I_B[k+j-1]| \le \frac{\bar{U}_B}{\lambda_B}\left(1 - e^{-\lambda_B T_s}\right), \label{eq:MPC_slew}\\
& -\overline{I}_{S} \leq -I_B[k+j]-\hat{d}[k+j] \leq \overline{I}_{S},\label{eq:MPC_assump_constraints}\\
& E_B[k+j] \in [\underline{E}_{B}+ \delta_B, \overline{E}_{B}- \delta_B], \\
& E_S[k+j] \in [\underline{E}_{S}+ \delta_S, \overline{E}_{S}- \delta_S],
\end{align}
\end{subequations}
Here the reference is the battery current, $r_{k+j} = I_B^{\mathrm{ref}}[k+j]$, and
$\hat{z}_{k+j}$ collects the untracked components; the explicit form of $\hat{f}$ is
given in~\eqref{eq:bmatrixf}.
The margins $\delta_B, \delta_S \ge 0$ in the SOC constraints absorb the accumulated prediction error $\eta_{k+j} := y_{k+j} - \hat y_{k+j|k}$ (each
one-step increment bounded by $\varepsilon_E$), so that
feasibility of the tightened (nominal) problem implies satisfaction of the
original limits for every admissible disturbance. Two standard choices
apply~\cite{Rawlings2017-vw}: open-loop tightening, $\delta_B(j) = \delta_S(j)
= j\,\varepsilon_E$, whose margins grow linearly with the horizon and may become
infeasible for long $N$; or tube MPC, which uses an ancillary feedback
$K \in \mathbb{R}^{1\times 2}$ to confine the error to the minimal robust
positive invariant set, giving horizon-independent constant
margins at the cost of one additional gain to tune.

\begin{corollary}[High-level verification]\label{cor:high}
The MPC satisfies Assumption \ref{asm:mismatch} as robust SOC constraints
    are tightened in~\eqref{eq:HESS_MPC}, ensuring
    $\YN$ absorbs the worst-case
    mismatch. Assumption~\ref{asm:mpc_stability} holds as the quadratic cost 
yields descent~\eqref{eq:descent_condition} with $\alpha_3(s) = Q\,s^2$ (scalar weight on $E_B$), and Lipschitz regularity~\eqref{eq:lipschitz} holds for quadratic $V_N^*$. Thus, from Corollary \ref{cor:high_level_contract}, $\Sigma_H \models \mathcal{C}_H = \bigwedge_k \mathcal{C}_H^k$ holds if equation \eqref{eq:verticalcomp} condition holds with the following $\varepsilon_T$:
\begin{align}
  \varepsilon_T=\sqrt{\tfrac{\bar c_V}{Q}}\cdot\sqrt{\tfrac{L_V}{Q}}\cdot\sqrt{\varepsilon_E}
  =\frac{\sqrt{\bar c_V\,L_V\,\varepsilon_E}}{Q}. 
    \label{eq:eT_quadratic}
\end{align}
\end{corollary}
\begin{proof}
With scalar weight $Q>0$ on $E_B$, write the stage cost as
$\ell(\hat y,r)=Q\,|E_B-E_B^{\mathrm{goal}}|^2$ and read the
sandwich~\eqref{eq:sandwich_bounds} and descent~\eqref{eq:descent_condition} bounds in the scalar
$E_B$-error: since $V_N^*(\hat y)\ge \ell(\hat y,r_0^*)\ge Q\,|E_B-E_B^{\mathrm{goal}}|^2$, the
lower bound is rigorous with $\alpha_1(s)=Q\,s^2$; the descent~\eqref{eq:descent_condition}
gives $\alpha_3(s)=Q\,s^2$, and the upper bound $\alpha_2(s)=\bar c_V s^2$ of Assumption~2 (i), with
$\bar c_V\le\lambda_{\max}(Q)/(1-\rho)$ under standard terminal conditions
(exponential controllability of rate $\rho\in(0,1)$). The supercapacitor
energy $E_S$ is not regulated (Remark~\ref{rem:superbatt}); its boundedness is
enforced by the soft constraints~\eqref{eq:HESS_soft_safety}, not by this Lyapunov argument.
Composing $\varepsilon_T=\alpha_1^{-1}(\alpha_2(\alpha_3^{-1}(L_V\varepsilon_E)))$
yields~\eqref{eq:eT_quadratic}.
\end{proof}

\subsection{Low-Level Subsystem Design ($\Sigma_L$)}
\label{sec:HESS_low_level}

\textbf{Tracking controller $\Sigma_{\mathrm{track}}$.}
The inner-loop stabilizes voltage and current dynamics to track the 
reference $r = [V_{\mathrm{nom}}, I_B^{\mathrm{ref}}]^\top$.

\emph{Battery current regulation.}
The proportional controller
\begin{align}
    u_B(t) = -\lambda_B (I_B(t) - I_B^{\mathrm{ref}}), \quad \lambda_B > 0
    \label{eq:HESS_uB}
\end{align}
yields exponential convergence:
\begin{align}
    I_B(t) = I_B^{\mathrm{ref}} + (I_B(t_0) - I_B^{\mathrm{ref}}) e^{-\lambda_B(t-t_0)}
    \label{eq:HESS_IB_solution}
\end{align}
and guarantees the tracking of $I_B^{\mathrm{ref}}$.

\emph{Voltage regulation.}
The voltage-current subsystem reduces to:
\begin{align}
    \dot{V}_{\mathrm{gr}} = \frac{1}{C_{\mathrm{bus}}}(I_S + \bar{d}(t)), \quad
    \dot{I}_S = u_S + w(t),
    \label{eq:HESS_reduced}
\end{align}
where $\bar{d}(t) := d(t) + I_B(t)$ is known (load plus settling battery 
current) and $w(t)$ is a bounded unknown disturbance. The supercapacitor 
regulates voltage via disturbance-canceling PD control
\begin{align}
    u_S = -C_{\mathrm{bus}}\,k_1\,(V_{\mathrm{gr}} - v) - k_2\,(I_S + \bar{d}) - \dot{\bar{d}},
    \label{eq:uS}
\end{align}
where $v(t)$ is the filtered reference from the ERG, and $k_1, k_2 > 0$ are design parameters. The three terms in~\eqref{eq:uS} admit a physical interpretation:
$-C_{\mathrm{bus}}\,k_1\,(V_{\mathrm{gr}} - v)$ is proportional feedback on the voltage error,
$-k_2\,(I_S + \bar{d})$ is derivative (damping) feedback since
$(I_S + \bar{d})/C_{\mathrm{bus}} = \dot{V}_{\mathrm{gr}}$,
and $-\dot{\bar{d}}$ cancels the disturbance rate
(requiring an estimate of $\dot{\bar{d}}$).
The uncompensated reference accelerations in $\psi_{\mathrm{dyn}}$
are handled by the ERG.

Define the tracking error $e(t) := [e_V,\ e_{\dot V}]^\top$ where $e_V := V_{gr}-v, e_{\dot V} := \dot V_{gr}-\dot v$.
The closed-loop error dynamics is
\begin{align}
    \dot{e} = A\,e + B\,w(t) - B_v\,\psi_{\mathrm{dyn}}(t),
    \nonumber
\end{align}
where
\begin{align}
    A = \begin{bmatrix} 0 & 1 \\ -k_1 & -k_2 \end{bmatrix},
    \quad
    B = \begin{bmatrix} 0 \\ \tfrac{1}{C_{\mathrm{bus}}} \end{bmatrix},
    \quad
    B_v = \begin{bmatrix} 0 \\ 1 \end{bmatrix},
    \nonumber
\end{align}
and the feedforward residual is
\begin{align}
    \psi_{\mathrm{dyn}}(t) := k_2\,\dot{v}(t) + \ddot{v}(t),
    \nonumber
\end{align}
with $c_{\psi}$ chosen to bound both the $k_{2}\dot v$ and $\ddot v$
contributions of $\psi_{\mathrm{dyn}}$, i.e.\
$\|B_{v}\psi_{\mathrm{dyn}}\|\le c_{\psi}\cdot(\text{threshold rate})$.
With the plant-level disturbance bound
$\|w\| \leq W_{\max}$, the disturbance channel
$B$ contributes
\begin{align}
    \|B\,w\| \leq \|B\|\,|w| \leq \tfrac{W_{\max}}{C_{\mathrm{bus}}}
    =: H_{\max}.
    \nonumber
\end{align}
\emph{Gain selection.}
Choose $k_1, k_2 > 0$ such that $A$ is Hurwitz with decay rate: $\lambda_e := \min|\mathrm{Re}(\mathrm{eig}(A))| > 0$.
We assume $\lambda_B > \lambda_e$: the battery current loop settles within a sampling interval (see \appref{app:mismatch}), while the supercapacitor voltage loop, carrying the slower residual regulation, decays at the smaller rate $\lambda_e$.

\textit{ERG design.}
Each constraint in~\eqref{eq:HESS_safety_spec} is expressed in terms of 
error state $e$ and reference $v$:
\begin{align}
    c_{a,i}^\top e + c_{b,i}^\top \dot{e} \leq d_i(v),
    \label{eq:HESS_constraint_form}
\end{align}
where vectors $c_{a,i}, c_{b,i}$ and bound $d_i(v)$ are derived for each 
constraint. $\Gamma(v)$ incorporates input constraints, decreasing near boundaries 
to ensure feasible control effort (see~\appref{app:constraints}).

\begin{proposition}[Low-level verification.]\label{cor:low}
With the low-level design $\Sigma_L$ in Section \ref{sec:HESS_low_level} and low-level contract $\mathcal{C}_L$ in \eqref{eq:lowlevelc}
with following $\bar{r}$ and $\varepsilon_L$, under timing compatibility condition $C_{tss}$ in Proposition \ref{prop:settling}, $\Sigma_L \models \mathcal{C}_L$ holds. Moreover, the input constraints hold $|u_B(t)| \leq \bar{U}_B$ and $|u_S(t)| \leq \bar{U}_S$ for all $t \in [t_k, t_{k+1})$:
\begin{align}
\bar{r} &:= [0,\bar{r}_B] \\
\bar{r}_B &:= \frac{\bar{U}_B}{\lambda_B}(1 - e^{-\lambda_B T_s})\\
\varepsilon_L &:= [\varepsilon_{L,V} := \varepsilon_V ,\varepsilon_{L,I_B} := \frac{\bar{U}_B}{\lambda_B}e^{-\lambda_B T_s}]\label{eq:varepsilon1}
\end{align}  
where $\varepsilon_{L,V} := \varepsilon_V$ in \eqref{eq:eps_L}, i.e., $\sqrt{\Vbar(H_{\max}) \cdot [P^{-1}]_{11}}$. The matrix $P$ is from quadratic Lyapunov function $V(e) = e^\top Pe$ that satisfies the Lyapunov equation: $A^\top P + PA = -R$, $R$ is a symmetric positive definite matrix. 
\end{proposition}

\begin{proof}
We separately consider the voltage-current linear subsystem \eqref{eq:HESS_reduced} and the battery subsystem \eqref{eq:HESS_IB_solution}.

\textbf{(The voltage-current subsystem \eqref{eq:HESS_reduced})} Assumption \ref{asm:quad_iss} holds as we choose $k_1, k_2 > 0$ such that $A$ is Hurwitz. 
Let $H_{\max}$ satisfy Assumption~\ref{asm:adm_dist}. Assumptions~\ref{asm:smoothing}--\ref{asm:progress} (negligible smoothing, 
reference admissibility, guaranteed progress) hold by design: 
$\eta \ll \bar{r}$, MPC incorporates constraint margins, and 
$\delta_{\mathrm{rep}} < 1$.
Then, from Corollary \ref{cor:lowlevel}, the ultimate bound in \eqref{eq:eps_L} is instantiated 
to the tracking bound for the voltage (1st element) of the voltage-current subsystem as \eqref{eq:varepsilon1}.

\textit{Input constraint $|u_S(t)| \leq \bar{U}_S$.} Refer to \appref{app:constraints}. 

\textbf{(The battery subsystem \eqref{eq:HESS_IB_solution})} \textit{Input constraint $|u_B(t)| \leq \bar{U}_B$.} During $\forall t \in [t_k, t_{k+1})$, 
\begin{align}
    |u_B(t)| = \lambda_B\,|I_B(t) - I_B^{\mathrm{ref}}[k]| \leq \lambda_B \cdot \frac{\bar{U}_B}{\lambda_B} = \bar{U}_B.
\end{align}
Since the $A_{\mathrm{ref}}^k$ holds for all $k$, the input constraint is satisfied for all $t \geq 0$.
\end{proof}

\begin{remark}
The timing compatibility condition $C_{tss}$ can be satisfied either by aggressive tuning of the low-level gains ($k_1, k_2$) to reduce the settling time $\tau_{LL}$, or by selecting a sufficiently long high-level sampling period $T_s$.
\end{remark}

Now, we verify the upward refinement condition in Definition~\ref{def:vertical_refinement} by deriving a bound on the discrete energy mismatch. 

\begin{proposition}[Vertical refinement condition]
\label{prop:mismatch_decomposed}
Let $\Sigma_H \triangleright\Sigma_L$ designed as in Corollaries \ref{cor:lowlevel} and \ref{cor:high}. Then, the refinement condition in Definition~\ref{def:vertical_refinement} holds at each $k$. In particular, the upward condition is derived as a discrete energy mismatch:
\begin{align}
\label{eq:affine_bound}
\varepsilon_{E}(\varepsilon_{L})
:=\bigl\|\,\underbrace{\Delta_{\mathrm{tr}}}_{\text{Transit Cost}}
+\underbrace{d_{\mathrm{ss}}\cdot\tau_{2}}_{\text{Settling Cost}}\,\bigr\|,
\end{align}
where $\Delta_{\mathrm{tr}},\,d_{\mathrm{ss}}\in\mathbb{R}^{2}$ and $\|\cdot\|$ is the same norm used in $A_{\mathrm{mis}}^{k}$~\eqref{eq:A_mis},
aggregating the per-channel (battery, supercapacitor) energy-mismatch contributions. $\Delta_{\mathrm{tr}} \in \mathbb{R}^2$ represents the fixed energy cost of the worst-case maneuver (transit and acceleration), and $\mathbf{d}_{\mathrm{ss}} \in \mathbb{R}^2$ represents the steady-state drift rate due to persistent environmental disturbances.\footnote{The explicit derivation of these coefficients is provided in \appref{app:mismatch}.}
\end{proposition}

\begin{proof} We prove the downward and upward handshake for both battery and supercapacitor-related terms.

\textbf{(Downward handshake)} We establish the reference-increment assumption $A_{\mathrm{ref}}^{k}$~\eqref{eq:A_ref}
for both channels. For the voltage channel, the grid target is held constant at
$r_V = V_{\mathrm{nom}}$, so $r_k - r_{k-1} = 0$ and $A_{\mathrm{ref}}^{k}$ holds
trivially. For the battery channel, $A_{\mathrm{ref}}^{k}$ holds for all $k \in \mathbb{N}$
by induction on $k$.

Assume the $A_{\mathrm{ref}}^k$ holds at $t_k$:
\begin{equation}\label{eq:inductive_hyp}
    |I_B(t_k) - I_B^{\mathrm{ref}}[k]| \;\leq\; \frac{\bar{U}_B}{\lambda_B}.
\end{equation}
On the interval $[t_k, t_{k+1})$, the reference is constant at $I_B^{\mathrm{ref}}[k]$ and the closed-loop error decays exponentially. At $t_{k+1} = t_k + T_s$:
\begin{align}\label{eq:decay}
    |I_B(t_{k+1}) - I_B^{\mathrm{ref}}[k]| \;&=\; |I_B(t_k) - I_B^{\mathrm{ref}}[k]|\,e^{-\lambda_B T_s} \nonumber\\
    &\leq\; \frac{\bar{U}_B}{\lambda_B}\,e^{-\lambda_B T_s}. \nonumber
\end{align}
At $t_{k+1}$ the reference jumps to $I_B^{\mathrm{ref}}[k+1]$. By the triangle inequality, $A_{\mathrm{ref}}^{k+1}$ holds:
\begin{align}
    &|I_B(t_{k+1}) - I_B^{\mathrm{ref}}[k+1]|\nonumber \\
    &\;\leq\; |I_B(t_{k+1}) - I_B^{\mathrm{ref}}[k]| + |I_B^{\mathrm{ref}}[k] - I_B^{\mathrm{ref}}[k+1]| \nonumber\\
    &\;\leq\; \frac{\bar{U}_B}{\lambda_B}\,e^{-\lambda_B T_s} + \frac{\bar{U}_B}{\lambda_B}\left(1 - e^{-\lambda_B T_s}\right) =\; \frac{\bar{U}_B}{\lambda_B}
\end{align}

\textbf{(Upward handshake)} 
The voltage tracking guarantee $G_{\mathrm{track}}^{k-1}$, i.e.
$| V_{\mathrm{gr}}(t_k) - V_{\mathrm{nom}} \bigr| \le \varepsilon_V \text{ at } t_k$, follows from settling over the interval $k-1$ (Corollary~\ref{cor:lowlevel}
under $C_{\mathrm{tss}}$) and supplies the voltage component of the mismatch
bound $\tilde{w}_k$ feeding $A_{\mathrm{mis}}^{k}$. The element-wise bounds are
derived in \appref{app:mismatch}.
\end{proof}

\begin{corollary}
Under the same condition as in Proposition \ref{prop:mismatch_decomposed}, if the compatibility conditions hold, i.e., $\varepsilon_H$ satisfies $\varepsilon_H \ge \varepsilon_T+\delta$ in \eqref{eq:verticalcomp} with $\varepsilon_L$ in Corollary \ref{cor:lowlevel} and $\varepsilon_E$ in \eqref{eq:affine_bound}.
Then, from Theorem \ref{thm:layered_correctness}, the HESS specifications are satisfied by $\Sigma_H \triangleright \Sigma_L$.

\end{corollary}





\section{Numerical Validation}
\label{sec:simulation}

We validate the proposed hierarchical architecture through a numerical
simulation of the HESS case study. The system parameters represent a
grid-connected storage unit subject to physical constraints: voltage limits
$\mathcal{R}_V = [380, 420]\,$V, supercapacitor current limits
$\mathcal{R}_{I_S} = [-12, 12]\,$A, input saturation $U_{S,\max} = 720\,$A/s,
and bus capacitance $C_{\mathrm{bus}} = 1$. The simulation activates the full
hierarchy---MPC planner, ERG, and low-level controller---and verifies both
safety ($G_{\text{safe}}$) and liveness ($G_{\text{live}}$) under a
time-varying load. As a preliminary check, the low-level guarantees in
isolation---robust control invariance of $\Omega_h$ and the ISS decay bound
$\tau_2$---were verified separately through a stationary-reference test; the
phase portrait and ISS-envelope verification are reported in
\appref{app:rci}. The methodology used below to compute $\bar V_h$,
$\varepsilon_{L,V}$, and $\tau_2$ is the one developed there.

\emph{Controller and certificates.} The controller gains are $k_1 = 560$,
$k_2 = 48$, with closed-loop eigenvalues $\lambda_{1,2} = -20.0, -28.0$ and
decay rate $\lambda_e = 20.0~\mathrm{s}^{-1}$---a four-fold speed-up of the
low-level loop chosen so that the settling time clears the sampling period
(see \emph{Load and sampling} below). The ISS-Lyapunov function
$V(e) = e^\top P e$ is obtained from $A^\top P + P A = -R_L$ with
$R_L = \mathrm{diag}(100, 10)$, yielding
\begin{equation}
P = \begin{bmatrix} 63.7 & 0.089 \\ 0.089 & 0.106 \end{bmatrix},
\quad \kappa(P) = \frac{\lambda_{\max}(P)}{\lambda_{\min}(P)} = 601 .
\label{eq:P_scenarioB}
\end{equation}
Under the residual disturbance bound $\|w\| \le W_{\max} = 3.0~\mathrm{A/s}$
(so that $H_{\max} = W_{\max}/C_{\mathrm{bus}}$), the ISS ultimate sublevel
threshold $\bar V_h$ in~\eqref{eq:Vbar_def}---computed as in
\cite[Lemma~3]{NicotraUAVERG} (see \appref{app:rci})---is
$\bar V_h = 0.011~\mathrm{V}^2$, with induced per-coordinate voltage-error
bound $\varepsilon_{L,V} = \sqrt{\bar V_h\,[P^{-1}]_{11}} = 0.013~\mathrm{V}$
(cf.~\eqref{eq:eps_L}). The faster loop tightens this tracking tolerance by an
order of magnitude, but it also raises the actuator demand: the input
bottleneck $\Gamma_u = U_{S,\max}^2/(K^\top P^{-1}K)$
(with $K=[C_{\mathrm{bus}}k_1, C_{\mathrm{bus}}k_2]^\top$,
cf.~\appref{app:constraints}) scales as $\|K\|^2$, so restoring headroom
for the off-target start requires scaling the supercapacitor slew authority
with the loop bandwidth, here $U_{S,\max} = 720~\mathrm{A/s}$. With this choice
$\Gamma_u = 20.0$ and the binding constraint shifts to the supercapacitor
current limit, giving $\Gamma_{\min} = 12.1$.

\emph{Load and sampling.} The feedforward term $-C_{\mathrm{bus}}^{-1}\dot{\bar d}$
requires $\dot{\bar d}$ to be bounded, so the external load ramps from $0$ to
$-5$~A over $[0.5, 0.8]$~s (with a superimposed $2$~Hz oscillation) rather than
stepping. This decomposition is central to the hierarchical contract: the MPC
plans the battery current $I_B$ to shape $\bar d$, the feedforward in $u_S$
cancels $\bar d$ from the error dynamics, and the ERG together with the robust
control invariant (RCI) set $\Omega_h = \{e : V(e) \le \bar V_h\}$ handle only
the unknown residual $w$. The MPC samples at $T_s = 0.5$~s with a receding
horizon of $N = 20$ (a $10$~s prediction horizon). Under the feedforward
cancellation in \eqref{eq:uS} the transit phase is negligible
($\tau_1 \approx 0$), so $\tau_{LL} \approx \tau_2$, and the settling-time
certificate of \appref{app:rci} (Fig.~\ref{fig:app_settling}) gives
$\tau_{LL} = 0.43$~s. The timing-compatibility condition
$C_{\mathrm{tss}}: T_s \ge \tau_{LL}$ is therefore \emph{satisfied} with a
$13\%$ margin. 

\emph{Recursive feasibility.} At $T_s = 0.5$~s the one-step prediction errors
of the abstract model~\eqref{eq:abstract_model} are strongly asymmetric across
channels: $\varepsilon_{E_B} \approx 0.06~\mathrm{A\cdot s}$ for the battery
versus $\varepsilon_{E_S} \approx 1.75~\mathrm{A\cdot s}$ for the
supercapacitor. 
This is structural (cf.~Remark~\ref{rem:superbatt}): the tracked
battery channel ($\lambda_B \gg \lambda_e$) renders $\hat E_B[k+1]$
accurate, while the untracked supercapacitor buffer absorbs the
residual $w$ and intra-sample load variation and is poorly predicted
(the gap is robust across $T_s$).
Consequently
the supercapacitor SOC bound cannot be enforced as a hard constraint---its
nominal trajectory drifts over the horizon---which is exactly why the SOC
bounds~\eqref{eq:HESS_soft_safety} are treated as soft. On the battery channel, where
the SOC bound is tightened (Assumption~\ref{asm:mismatch}), the small
$\varepsilon_{E_B}$ yields a terminal open-loop margin
$\delta_B(N{-}1) = (N{-}1)\varepsilon_{E_B} = 1.2~\mathrm{A\cdot s}$, leaving the
tightened band open around the target; an LP feasibility check confirms the
problem~\eqref{eq:mpc_ocp} stays feasible step-to-step under worst-case
$\|\tilde w\| \le \varepsilon_E$. Robust recursive feasibility therefore holds
at $N=20$ without resort to tube tightening, the large supercapacitor mismatch
being confined to the soft channel.

\emph{Safety.} Figure~\ref{fig:scenario_B}(a) shows the bus voltage zoomed to
the actual fluctuation range. With feedforward active,
$V \in [400.0, 400.3]$~V---the known disturbance is almost perfectly cancelled
and the tracking error is dominated by the residual $w$. Panel~(c) confirms the
ERG safety invariant: $V(e(t)) < \Gamma(v(t))$ at every instant, with
\begin{equation}
V(e)_{\max} = 5.73 \;<\; \Gamma_{\min} = 12.1 ,
\end{equation}
so the safe set is never violated during convergence from the off-target
initial condition ($V(0) = 400.3$~V). After the transient, the error settles
into $\Omega_h$ ($V(e)\le\bar V_h=0.011\ll\Gamma$): once the feedforward has
canceled the known load, the error dynamics see essentially only the residual
$w(t)$, and the ERG's role reduces to certifying $V(e)<\Gamma(v)$ continuously even though the MPC updates only every $T_s$.
The realized $\Omega_h$-entry at $t\approx0.23$~s is faster than the ISS-envelope
certificate $\tau_{LL}=0.43$~s, so the certificate is met rather than exceeded.

\emph{Liveness.} Panel~(b) of Figure~\ref{fig:scenario_B} shows the current
allocation: the MPC commands $I_B$ to accumulate charge while the
supercapacitor $I_S$ absorbs the residual fast dynamics within
$\pm 12$~A. Panel~(d) shows the battery charge $E_B(t)$ converging to the
$5$\,A$\cdot$s target by $t\approx 3.8$\,s, verifying the liveness guarantee
$G_{\mathrm{live}}$.

\begin{figure}[htbp]
\vspace{-0.5cm}
	\centering
    \includegraphics[width=\linewidth]{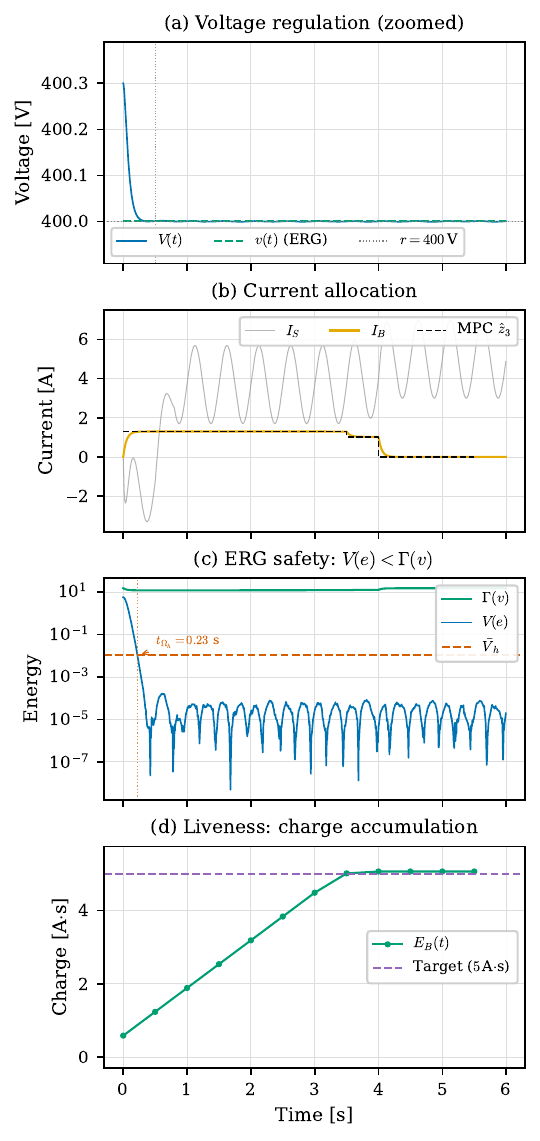}
	\caption{Full hierarchical operation under time-varying load with feedforward
cancellation, at the fast-loop tuning ($\lambda_e=20$, $T_s=0.5$~s).
(a)~Voltage regulation (zoomed): $V\in[400.0,400.3]\,$V.
(b)~Current allocation: MPC battery command $\hat z_3$ vs.\ supercapacitor fast
response. (c)~ERG safety margin: $V(e)<\Gamma(v)$ throughout, with
$\bar V_h\ll\Gamma(v)$; the trajectory enters $\Omega_h$ at $t\approx0.23$\,s,
well inside the certificate $\tau_{LL}=0.43$\,s. (d)~Battery charge $E_B(t)$
converges to the $5$\,A$\cdot$s target (within tolerance $\varepsilon_H$),
verifying $G_{\mathrm{live}}$.}
	\label{fig:scenario_B}
\end{figure}


\section{Conclusion}
\label{sec:conclusion}

This paper presented a heterogeneous assume--guarantee contract framework for layered control, formalizing the interface between a discrete-time planning layer and a continuous-time safety layer by introducing vertical refinement and timing compatibility conditions. We then provide the implementation of this method using an ERG-Tracker combination in the safety layer. Applied to a HESS case study, the framework yields a clean separation of concerns: the MPC plans slow energy states, the feedforward controller cancels known disturbances from the error dynamics, and the ERG certifies constraint satisfaction in continuous time. Numerical validation confirmed that safety is maintained and liveness is achieved. A gap remains between the ISS-based settling time bound and the actual convergence rate under feedforward cancellation; deriving a tighter certificate is a direction for future work.

\bibliographystyle{unsrt}
\bibliography{library}

\newpage

\appendix

\textbf{Suplemental Materials}

\section{Proof of Lemma \ref{lem:dissipation}}
By Assumption~\ref{asm:mismatch}, the perturbed
successor satisfies $\yhat_{k} \in \YN$, so
the MPC problem remains feasible, and all bounds in
Assumption~\ref{asm:mpc_stability} apply at
$t_{k}$.  Decompose:
\begin{align}
  &V_N^*(\hat y_{k|k}) - V_N^*(\hat y_{k-1|k-1}) \notag \\
    &= V_N^*\!\big(\hat f(\hat y_{k-1|k-1},\hat z_{k-1|k-1},r_{k-1})\big)
       - V_N^*(\hat y_{k-1|k-1}) \notag\\
    &\quad + V_N^*(\hat y_{k|k})
       - V_N^*\!\big(\hat f(\hat y_{k-1|k-1},\hat z_{k-1|k-1},r_{k-1})\big).
    \label{eq:decomposition}
\end{align}
Adding the two bounds~\eqref{eq:descent_condition} and~\eqref{eq:lipschitz} yields~\eqref{eq:dissipation}.

\section{HESS Constraint Reformulation}
\label{app:constraints}

\textbf{(Voltage state constraints)} The bus voltage must satisfy $V_{\min} \leq V_{\text{gr}} \leq V_{\max}$. Substituting $V_{\text{gr}} = e_V + v$:
\begin{align}
V_{\min} \leq e_V + v \leq V_{\max}.\nonumber
\end{align}
This yields two half-space constraints:
\begin{align}
\label{eq:constraint_vmin}
-e_V \le v - V_{min}, \\
\label{eq:constraint_vmax}
e_V &\leq V_{\max} - v. 
\end{align}

\textbf{(Supercap state constraint $|I_S| \leq \overline{I}_{S}$)} The supercapacitor current must satisfy $|I_S| \leq \overline{I}_{S}$. From the bus voltage dynamics \eqref{eq:HESS_reduced}:
\begin{align}
I_S = C_{\mathrm{bus}} \dot{V}_{\text{gr}} - \bar{d} = C_{\mathrm{bus}} (\dot e_V + \dot{v}) - \bar{d},
\end{align}
where $\bar{d} = I_B + d$. The constraint becomes:
\begin{align}
|C_{\mathrm{bus}} (\dot e_V + \dot{v}) - \bar{d}| \leq \overline{I}_{S}.
\end{align}
A sufficient condition using the triangle inequality is:
\begin{align}
|C_{\mathrm{bus}} \dot e_V| + |C_{\mathrm{bus}} \dot{v}| + |\bar{d}| \leq \overline{I}_{S}.
\end{align}
Rearranging for the error-dependent term:
\begin{align}
|C_{\mathrm{bus}} \dot e_V| \leq \overline{I}_{S} - |\bar{d}| - |C_{\mathrm{bus}} \dot{v}|.
\end{align}
Substituting worst-case bounds $|\dot{v}| \leq \bar{\kappa} \Gamma(v)$ and $|\bar{d}| \leq \bar{d}_{\max}$:
\begin{align}
|C_{\mathrm{bus}} \dot e_V| \leq \overline{I}_{S} - \bar{d}_{\max} - C_{\mathrm{bus}} \bar{\kappa} \Gamma(v).
\end{align}
This yields two half-space constraints:
\begin{align}
\label{eq:constraint_Is_upper}
C_{\mathrm{bus}} \dot e_V &\leq \bar{I}_S - C_{\mathrm{bus}} \bar{\kappa} \Gamma(v), \\
\label{eq:constraint_Is_lower}
-C_{\mathrm{bus}} \dot e_V &\leq \bar{I}_S - C_{\mathrm{bus}} \bar{\kappa} \Gamma(v),
\end{align}
where $\bar{I}_S := \overline{I}_{S} - \bar{d}_{\max}$ is the effective current margin after reserving capacity for the disturbance. 

\textbf{(Input constraint)} The supercapacitor control input $u_S$ must satisfy $|u_S| \leq U_{S,\max}$. From the control law \eqref{eq:uS}:
\begin{equation}
  u_S = -C_{\mathrm{bus}}k_1 e_V - k_2 C_{\mathrm{bus}}(\dot e_V +\dot v) - \dot{\bar d}(t).
  \label{eq:uS-appB}
\end{equation}
Applying the triangle inequality with worst-case bounds
$|\dot{v}| \leq \bar{\kappa}\Gamma(v)$ and $|\dot{\bar{d}}| \leq |\dot{\bar{d}}|_{\max}$, a sufficient condition is:
\begin{equation}
  |\tilde k_1 e_V + \tilde k_2 \dot e_V| \le \bar U_S - \tilde k_2 \bar\kappa\,\Gamma(v),
\end{equation}
where $\tilde k_1 := C_{\mathrm{bus}}k_1$, $\tilde k_2 := k_2 C_{\mathrm{bus}}$,
and $\bar U_S := U_{S,\max} - |\dot{\bar d}|_{\max}$. This yields
\begin{align}
  \tilde k_1 e_V + \tilde k_2 \dot e_V &\le \bar U_S - \tilde k_2 \bar\kappa\,\Gamma(v),
  \label{eq:constraint_sc_upper}\\
  -\tilde k_1 e_V - \tilde k_2 \dot e_V &\le \bar U_S - \tilde k_2 \bar\kappa\,\Gamma(v).
  \label{eq:constraint_sc_lower}
\end{align}
Note that the battery's input constraint \eqref{eq:HESS_input_constraints} is handled by the MPC slew rate \eqref{eq:MPC_slew}.  

The derived constraints are all in the form \eqref{eq:HESS_constraint_form}.

\section{Mismatch Bound Derivation}
\label{app:mismatch}

We derive the prediction mismatch bound in Proposition~\ref{prop:mismatch_decomposed}.

\textbf{Battery energy prediction error.}
The abstract model predicts:
\begin{align}
    \hat{E}_B[k+1] = E_B[k] + T_s c_B V_{\mathrm{nom}} I_B^{\mathrm{ref}}[k]
\end{align}

The physical evolution is:
\begin{align}
    E_B(t_{k+1}) = E_B(t_k) + \int_{t_k}^{t_{k+1}} c_B V_{\mathrm{gr}}(t) I_B(t) \, dt
\end{align}

The mismatch is:
\begin{align}
    \tilde{w}_{k,E_B} &= E_B(t_{k+1}) - \hat{E}_B[k+1] \\
    &= c_B \int_{t_k}^{t_{k+1}} \left( V_{\mathrm{gr}}(t) I_B(t) - V_{\mathrm{nom}} I_B^{\mathrm{ref}}[k] \right) dt
    \label{eq:mismatch_EB}
\end{align}

\textbf{Decomposition.}
Define deviations:
\begin{align}
    \tilde{V}(t) &:= V_{\mathrm{gr}}(t) - V_{\mathrm{nom}} \\
    \tilde{I}_B(t) &:= I_B(t) - I_B^{\mathrm{ref}}[k]
\end{align}

Expanding:
\begin{align}
    V_{\mathrm{gr}} I_B - V_{\mathrm{nom}} I_B^{\mathrm{ref}} 
    = \tilde{V} I_B^{\mathrm{ref}} + V_{\mathrm{nom}} \tilde{I}_B + \tilde{V} \tilde{I}_B
\end{align}

\textbf{Phase-wise bounds.}

\emph{Transit phase} ($t \in [t_k, t_k + \tau_1]$):
\begin{align}
    |\tilde{V}(t)| &\leq e_{\mathrm{peak}} + \eta && \text{(error + ERG lag)} \\
    |\tilde{I}_B(t)| &\leq \frac{\bar{U}_B}{\lambda_B} e^{-\lambda_B(t-t_k)} && \text{(exponential settling)}
\end{align}
Integrating:
{\small
\begin{align}
\Delta_{\mathrm{tr},B} := c_B\Big[\,\bar I_B\,e_{\mathrm{peak}}\,\tau_1
\;+\;(V_{\mathrm{nom}}+e_{\mathrm{peak}}+\eta)\,\tfrac{\bar U_B}{\lambda_B}\big(1-e^{-\lambda_B\tau_1}\big)\Big].
\end{align}
}
\emph{Settling phase} ($t \in [t_k + \tau_1, t_{k+1}]$):
\begin{align}
    |\tilde{V}(t)| &\leq \varepsilon_V + \delta + \eta && \text{(settled error)} \\
    |\tilde{I}_B(t)| &\approx 0 && \text{(battery settled)}
\end{align}
The bound on $|\tilde{I}_B(t)|$ is obtained as using $\lambda_B >> \lambda_e$. The battery is dedicated to the energy target and thus as soon as it observes a new target at each time step, it tracks that target (cf. Remark \ref{rem:superbatt}).
The drift rate:
\begin{align}
    d_{\mathrm{ss},B} := \lambda_B \bar{I}_B ((1+\delta)\varepsilon_V + \eta)
\end{align}

\textbf{Total bound.}
\begin{align}
    |\tilde{w}_{k,E_B}| \leq \Delta_{\mathrm{tr},B} + \mathbf{d}_{\mathrm{ss},B} \cdot \tau_2
\end{align}

\textbf{Supercapacitor energy prediction error.}
The planning model assumes an algebraic bus ($\dot{V}_{\mathrm{gr}} = 0$,\; $V_{\mathrm{gr}} = V_{\mathrm{nom}}$),
so that $I_S^{\mathrm{ref}}[k] = -\bar{d}[k] = -I_B[k] - \hat{d}[k]$. The abstract model predicts:
\begin{align}
    \hat{E}_S[k+1] = E_S[k] + T_s\, c_S\, V_{\mathrm{nom}}\, I_S^{\mathrm{ref}}[k].
\end{align}
The physical evolution is:
\begin{align}
    E_S(t_{k+1}) = E_S(t_k) + \int_{t_k}^{t_{k+1}} c_S\,V_{\mathrm{gr}}(t)\,I_S(t)\,dt.
\end{align}
The mismatch is:
\begin{align}
    \tilde{w}_{k,E_S} &= E_S(t_{k+1}) - \hat{E}_S[k+1] \\
    &= c_S \int_{t_k}^{t_{k+1}}
   \bigl(V_{\mathrm{gr}}(t) I_S(t) - V_{\mathrm{nom}} I_S^{\mathrm{ref}}[k]\bigr)\,dt.
    \label{eq:mismatch_ES}
\end{align}

\textbf{Decomposition.}
Define deviations:
\begin{align}
    \tilde{V}(t) &:= V_{\mathrm{gr}}(t) - V_{\mathrm{nom}} = e_V(t) + (v(t) - V_{\mathrm{nom}}), \\
    \tilde{I}_S(t) &:= I_S(t) - I_S^{\mathrm{ref}}[k].
\end{align}
From the bus dynamics $I_S = C_{\mathrm{bus}}\dot{V}_{\mathrm{gr}} - \bar{d}$
and the planning-model Definition~$I_S^{\mathrm{ref}}[k] = -\bar{d}[k]$,
the current deviation (neglecting disturbance drift) is:
\begin{align}
    \tilde{I}_S(t) = C_{\mathrm{bus}}\,\dot{V}_{\mathrm{gr}}(t)
    = C_{\mathrm{bus}}\bigl(\dot{v}(t) + e_{\dot V}(t)\bigr).
    \label{eq:IS_deviation}
\end{align}
The integrand expands as:
\begin{align}
    V_{\mathrm{gr}}\,I_S - V_{\mathrm{nom}}\,I_S^{\mathrm{ref}}
    = \tilde{V}\,I_S^{\mathrm{ref}} + V_{\mathrm{nom}}\,\tilde{I}_S + \tilde{V}\,\tilde{I}_S.
\end{align}

\textbf{Phase-wise bounds.}

\emph{Transit phase} ($t \in [t_k,\, t_k + \tau_1]$):
The reference $v(t)$ moves from $r_{k-1}$ toward $r_k = V_{\mathrm{nom}}$,
so $\dot{v} \neq 0$ and the capacitor carries current
even in the ideal case.
The ISS bound on the error state
$e = [e_V,\, e_{\dot V}]^\top$ in the $\ell_\infty$ norm gives:
\begin{align}
    |\tilde{V}(t)| &\leq e_{\mathrm{peak}} + \eta\\
    |\tilde{I}_S(t)| &\leq C_{\mathrm{bus}}\bigl(\kappa_{\max} + e_{\mathrm{peak}}\bigr)
\end{align}
where $\kappa_{\max} = \bar{\kappa}\Gamma(v)$
is the peak ERG reference rate in Assumption \ref{asm:progress}.
Since each bound is pointwise, integration gives:
\begin{align}
    \Delta_{\mathrm{tr},S} &:= c_S \bar{I}_S (e_{\mathrm{peak}} + \eta)\,\tau_1
   \\
     &+ (V_{\mathrm{nom}} + e_{\mathrm{peak}} + \eta)\bigl(C_{\mathrm{bus}}\kappa_{\max} + e_{\mathrm{peak}}\bigr)\tau_1.
\end{align}

\emph{Settling phase} ($t \in [t_k + \tau_1,\, t_{k+1}]$):
The reference has reached $v = V_{\mathrm{nom}}$, so $\dot{v} = 0$.
The current deviation reduces to:
\begin{align}
    \tilde{I}_S(t) = C_{\mathrm{bus}}\,e_{\dot V}(t),
\end{align}
and the bounds become:
\begin{align}
    |\tilde{V}(t)| &\leq \varepsilon_V + \delta + \eta
    && \text{(settled voltage error)}, \\
    |\tilde{I}_S(t)| &\leq C_{\mathrm{bus}}\,\varepsilon_{\dot V}
    && \text{(settled velocity error)},
\end{align}
where $\varepsilon_V, \varepsilon_{\dot V}$ are the ultimate bounds on $|e_V|, |e_{\dot V}|$
from the ISS gain.
The drift rate:
\begin{align}
d_{\mathrm{ss},S} &:= c_S \bar{I}_S\bigl((1+\delta)\varepsilon_V + \eta\bigr)\\
   &+ (V_{\mathrm{nom}} + \varepsilon_V + \eta)\, C_{\mathrm{bus}} \varepsilon_{\dot V}.
\end{align}

\textbf{Total bound.}
\begin{align}
    |\tilde{w}_{k,E_S}| \leq \Delta_{\mathrm{tr},S}
    + \mathbf{d}_{\mathrm{ss},S} \cdot \tau_2.
\end{align}
Stacking the channel bounds,
\begin{align}
&\|\tilde{w}_{k}\|\le
\bigl\|[\,\tilde{w}_{k,E_{B}},\ \tilde{w}_{k,E_{S}}\,]^{\top}\bigr\|=\varepsilon_{E},
\\
&\Delta_{\mathrm{tr}}:=[\Delta_{\mathrm{tr},B},\,\Delta_{\mathrm{tr},S}]^{\top},
\quad
d_{\mathrm{ss}}:=[d_{\mathrm{ss},B},\,d_{\mathrm{ss},S}]^{\top}.
\end{align}

\section{Additional Simulation on Invariance}
\label{app:rci}

This appendix validates, in isolation, the two low-level guarantees that the
hierarchical contract relies on: (i) the robust control invariance (RCI) of
$\Omega_h = \{e : V(e) \le \bar V_h\}$, which the ERG requires for real-time
constraint enforcement; and (ii) the ISS decay bound, which certifies the
finite settling time $\tau_2$. These mechanisms serve different layers of the
contract and are verified separately below; the disturbance profile used here
isolates the low-level loop with the ERG reference held stationary
($\dot v = 0$). The controller gains are those of the deployed design in
Section~\ref{sec:simulation} ($k_1=560$, $k_2=48$, i.e. closed-loop eigenvalues
$-20,\,-28$ and decay rate $\lambda_e = 20~\mathrm{s}^{-1}$). To isolate the
low-level loop under a stationary ERG reference ($\dot v=0$), we replace the
time-varying load of Section~\ref{sec:simulation} with a lumped adversarial
disturbance capturing unmodeled dynamics, measurement noise, and load
fluctuations:
\begin{align}
	\label{eq:disturbance}
	w(t) = W_{\max}\bigl[0.7\sin(15t) + 0.3\,\xi(t)\bigr],
\end{align}
where $|\xi(t)|\le 1$ and $W_{\max}=3.0~\mathrm{A/s}$.

\subsubsection{RCI Set and ERG Safety Mechanism}
\label{sec:scenario_A_rci}

Following Lemma~3 of~\cite{NicotraUAVERG}, $\bar V_h$ is obtained by
maximizing $V(e)$ over the surface $\dot V(e,w) = 0$ under the worst-case
$|w| \le W_{\max}$. Parametrizing $e = a\,[\cos\theta, \sin\theta]^\top$, the
constraint $\dot V = 0$ fixes
\begin{equation}
a(\theta) = \frac{2\,|b^\top P B|\,H_{\max}}{b^\top R_L\,b},
\qquad b = [\cos\theta, \sin\theta]^\top,
\label{eq:app_atheta}
\end{equation}
and the invariant level is
\begin{equation}
\bar V_h = \max_{\theta \in [0, 2\pi)} a(\theta)^2 \, (b^\top P\,b)
= 0.011~\mathrm{V}^2 .
\label{eq:app_Vbar}
\end{equation}

The maximum is attained at $\theta^\ast = 248.0^\circ$ (equivalently $68.0^\circ$
by the $\pm$ symmetry of the $\dot V=0$ surface).
Figure~\ref{fig:phase_portrait} shows the phase portrait under
\eqref{eq:disturbance} with initial condition $e(0) = [3,0]^\top$: the trajectory
enters $\Omega_h$ at $t = 0.34$~s and remains inside thereafter, with zero
violations over the remaining simulation horizon.


\begin{figure}[htbp]
	\centering
	\includegraphics[width=\linewidth]{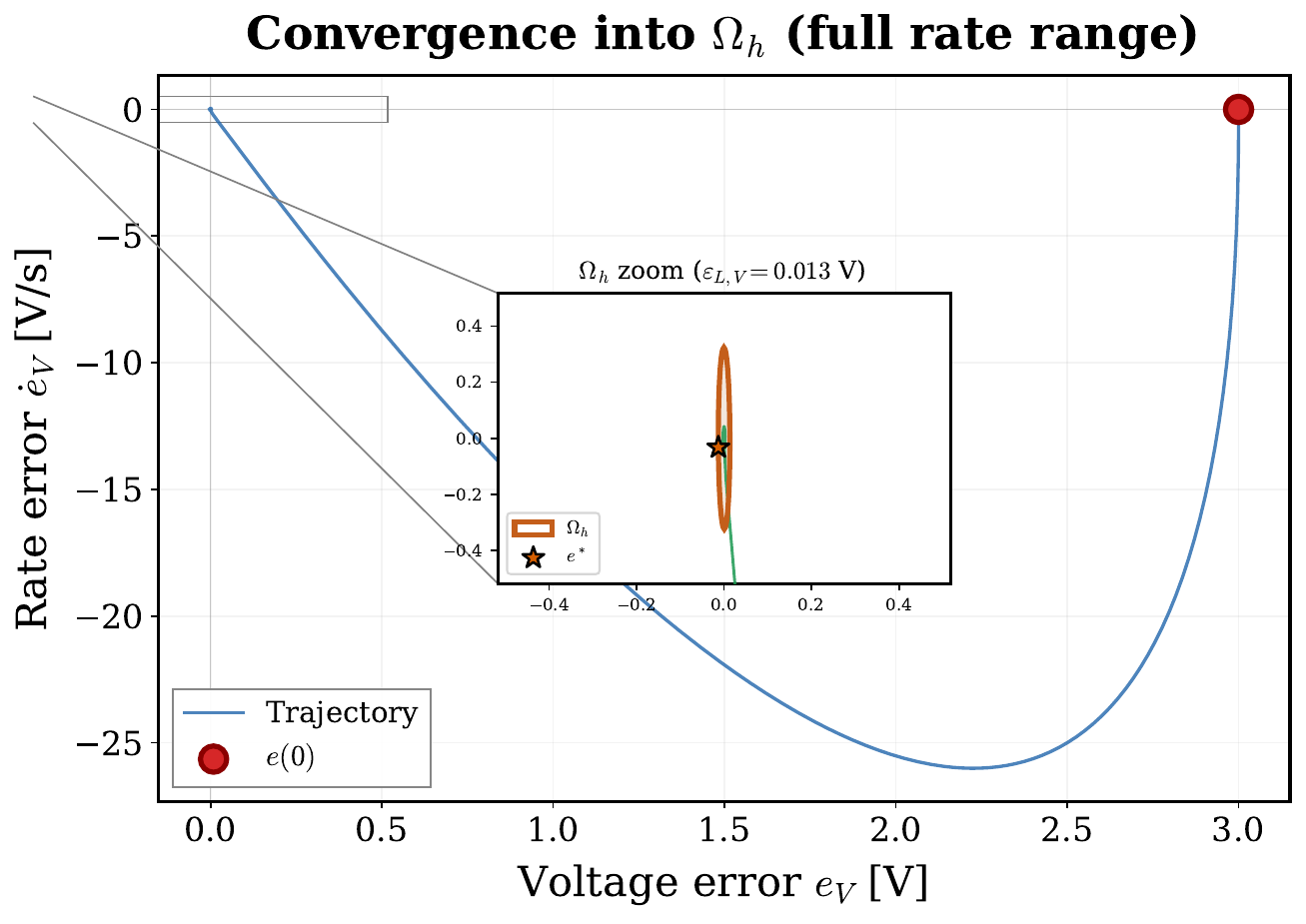}
	\caption{Phase portrait under mixed disturbance, full rate range. The trajectory
	converges from $e(0)=[3,\,0]^\top$ through a transient excursion in the rate
	coordinate $\dot e_V$ and settles into $\Omega_h$. \emph{Inset:} zoom on the
	invariant set $\Omega_h$ (shaded ellipsoid, $\varepsilon_{L,V}=0.013\,$V), with the
	worst-case point~$e^*$ on the $\dot V=0$ surface marked. Note the $\sim\!2000\times$
	scale separation between the transient and $\Omega_h$, a direct consequence of the
	fast low-level loop.}
	\label{fig:phase_portrait}
\end{figure}

\subsubsection{ISS Envelope and Decay Time}
\label{sec:scenario_A_iss}

For the settling-time analysis we require an ISS envelope bounding $\|e(t)\|$
during transients. Since $A$ is Hurwitz, $\|e^{At}\| \le m\,e^{-\lambda_e t}$
for some $m \ge 1$, giving
\begin{equation}
\|e(t)\|\le m\,e^{-\lambda_e t}\|e(0)\|+\varepsilon\,(1-e^{-\lambda_e t}),
\label{eq:app_iss_env}
\end{equation}
with ISS noise floor $\varepsilon = \gamma_{\mathrm{ISS}} H_{\max}$ and
$\gamma_{\mathrm{ISS}} = m/\lambda_e$. The overshoot factor $m$ is computed
in hindsight as the tightest value for which \eqref{eq:app_iss_env} holds over
the simulated trajectory, yielding $m = 29.2$, hence
$\gamma_{\mathrm{ISS}} = 29.2/20.0 = 1.46$ and
$\varepsilon = 1.46 \times 3.0 = 4.38~\mathrm{V}$. The larger $m$ relative to the
slower design of the original tuning reflects the non-normal overshoot of the
rate coordinate, which scales with $\lambda_e$. With tolerance $\delta = 0.1$
and peak excursion $\|e\|_{\mathrm{peak}} = m\|e(0)\| = 87.6~\mathrm{V}$,
\begin{align}
	\tau_2 = \frac{1}{\lambda_e}\ln\!\frac{m\,\|e\|_{\mathrm{peak}}}{\delta\,\varepsilon}
	= \frac{1}{20.0}\ln\!\frac{29.2\times 87.6}{0.1\times 4.38}
	= 0.43\,\mathrm{s}. \nonumber
\end{align}
Since the reference is held stationary ($\dot v=0$), the transit phase vanishes
($\tau_1=0$) and $\tau_{LL}=\tau_2=0.43~\mathrm{s}$. With the
Section~\ref{sec:simulation} sampling period $T_s = 0.5~\mathrm{s}$, the
timing-compatibility condition $C_{\mathrm{tss}}: T_s \ge \tau_{LL}$ is
\emph{satisfied} with a $13\%$ margin. The optimized invariant level
$\bar V_h = 0.011~\mathrm{V}^2$ yields component-wise bounds
$\varepsilon_{L,V} = 0.013~\mathrm{V}$ and $\varepsilon_{\dot V} = 0.32~\mathrm{V/s}$,
which serve as the ERG's constraint-tightening margins; the adversarial
simulation confirms forward invariance with substantial margin, the realized
worst case reaching only $V_{\max}^{\mathrm{adv}}/\bar V_h = 0.16$ of the
certified level.

Figure~\ref{fig:app_settling} confirms that the ISS envelope
\eqref{eq:app_iss_env} bounds $\|e(t)\|$ at every instant, and that
$\|e(t)\| \le (1+\delta)\varepsilon = 4.82~\mathrm{V}$ for all $t \ge \tau_{LL}$.
The settled error in fact lies far below this floor (within $\Omega_h$,
$\|e\|\lesssim 0.04~\mathrm{V}$): the $m$-based noise floor $\varepsilon$ is a
conservative envelope, not a tight steady-state bound, consistent with the gap
between the certificate and the realized convergence discussed in
Section~\ref{sec:conclusion}.

\begin{figure}[htbp]
	\centering
	\includegraphics[width=0.9\linewidth]{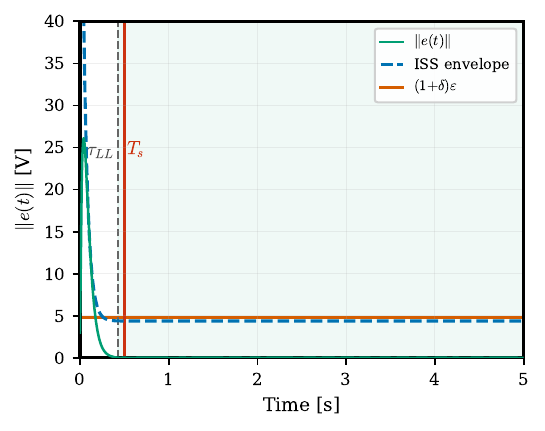}
	\caption{Error norm $\|e(t)\|$ (green), ISS envelope \eqref{eq:app_iss_env}
	(blue dashed), and the $(1{+}\delta)\varepsilon$ threshold (orange). The vertical
	markers locate the settling time $\tau_{LL}=0.43\,$s and the sampling period
	$T_s=0.5\,$s; since $\tau_{LL} < T_s$, the timing-compatibility condition
	$C_{\mathrm{tss}}$ holds. After $\tau_{LL}$ (shaded region) the bound is met with
	zero violations.}
	\label{fig:app_settling}
\end{figure}

Thus the isolated low-level loop meets both guarantees—RCI of
$\Omega_h$ and the settling certificate $\tau_{LL}\le T_s$—that the hierarchical contract relies on.
\section{Proof of Proposition~\ref{prop:settling}}\label{app:settling}

We bound the ISS envelope across the two sequential
phases of ERG operation
(cf.\ Fig.~\ref{fig:tracking_decomp}). Assumptions \ref{asm:smoothing}--\ref{asm:ref_adm} ensure the \emph{existence} of a feasible path to each target~$r_k$, which is naturally satisfied when the high-level MPC incorporates constraint margins.

\medskip
\emph{Phase 1: Transit
($t \in [t_k,\, t_k + \tau_1]$).}\;
By Assumption~\ref{asm:progress}, the ERG velocity
satisfies $\|\dot{v}\| \geq \underline{\kappa}\,
\underline{r}$ while
$v(t) \neq r_k$, so a step of size $\bar{r}$ is
traversed in at most
\begin{align}
    \tau_1
    \;\leq\;
    \tau_1^M := \frac{\bar{r}+\varepsilon}
         {\underline{\kappa}\,\underline{r}}.
    \label{eq:tau1}
\end{align}
Let $M$ be a uniform bound on
$\|g_e(\dot v, \ddot v, \ldots)\|$ along the ERG
flow; existence of such an $M$ follows from
standard smoothness arguments for navigation-based
reference governors~\cite{GaroneNicotraERG,NicotraUAVERG}. The total disturbance
during transit satisfies
\begin{align}
    D_{\mathrm{transit}}
    \;:=\;
    H_{\max} + M.
\end{align}
The ISS envelope gives
\begin{align}
    \|e(t)\|
    \;\leq\;
    m\,e^{-\lambda_e (t - t_k)}\,\|e(t_k)\|
    \;+\;
    \gamma_{\mathrm{ISS}}\,D_{\mathrm{transit}}.
    \label{eq:iss_envelope_transit}
\end{align}
Evaluating at $t = t_k + \tau_1$ with
$\|e(t_k)\| \leq \bar{r} + \varepsilon$
yields peak expression~\eqref{eq:tau_LL}.

\emph{Phase 2: Decay
($t \in [t_k + \tau_1,\, t_k + \tau_{LL}]$).}\;
After transit, the reference is stationary ($\dot v \approx 0$), so the
disturbance reduces to $H_{\max}$ and the envelope becomes
\begin{align}
    \|e(t)\|
    \;\leq\;
    m\,e^{-\lambda_e(t - t_k - \tau_1)}\,
    e_{\mathrm{peak}}
    \;+\; \varepsilon.
    \label{eq:decay_envelope}
\end{align}
Setting the transient term equal to $\delta\varepsilon$ and solving:
\begin{align}
  m\,e^{-\lambda_e \tau_2}\, e_{\mathrm{peak}} = \delta\,\varepsilon
  \quad\Longrightarrow\quad
  \tau_2 := \frac{1}{\lambda_e}
    \ln\!\Bigl(\frac{m\, e_{\mathrm{peak}}}{\delta\,\varepsilon}\Bigr).
\end{align}
For $t \geq t_k + \tau_1 + \tau_2 = t_k + \tau_{LL}$, the settling bound
\eqref{eq:tau_LL} follows, giving $\|e(t)\| \le (1+\delta)\varepsilon$.

\section{General Nonlinear Formulation of the Safety Layer}
\label{app:nonlinear}

This appendix develops the safety layer of Section~\ref{subsec:safety_layer} in its
general nonlinear form, of which the linear--quadratic body is the closed-form
specialization. The structural argument---unilateral safety via robust forward
invariance under a frozen-ISS reference governor---is identical; only the
$\mathcal{K}_\infty$ machinery is more general.

\subsection{Inner-Loop Error Dynamics}
Differentiating $e(t) := h_r(x(t)) - v(t)$ along~\eqref{eq:ct_system} and
substituting the tracker~\eqref{eq:tracking_controller} gives
$\dot e = P_r f(x,\kappa(x,v),w) - \dot v$. When $\kappa$ admits an
error-coordinate representation with additively-entering disturbance, this
becomes
\begin{equation}
\dot e = f_e(e) + B_w w + g_e(\dot v, \ddot v, \ldots),
\label{eq:app_err_dyn}
\end{equation}
where $f_e$ is the autonomous error evolution and $g_e$ the ERG-induced
reference coupling. The frozen dynamics ($\dot v = 0$) are
$\dot e = f_e(e) + B_w w$, with $\|B_w w\| \le H_{\max} := \|B_w\|\,W_{\max}$.

\subsection{General Frozen ISS}
\begin{assumption}[Frozen ISS]
\label{app:asm_frozen}
The frozen dynamics admit an ISS-Lyapunov function $V:\mathbb{R}^p\to\mathbb{R}_{\ge0}$
with $\underline\alpha, \overline\alpha, \alpha_d, \sigma \in \mathcal{K}_\infty$:
\begin{align}
\underline\alpha(\|e\|) \le V(e) \le \overline\alpha(\|e\|), \quad
\dot V(e) \le -\alpha_d(\|e\|) + \sigma(\|w\|).
\label{eq:app_iss_lyap}
\end{align}
\end{assumption}
Given $H_{\max} \ge 0$, define the ISS ultimate sublevel threshold
\begin{equation}
\bar V_h(H_{\max}) := \overline\alpha\big(\alpha_d^{-1}(\sigma(H_{\max}))\big).
\label{eq:app_Vbar_general}
\end{equation}
The sandwich~\eqref{eq:app_iss_lyap} ensures $V$ decreases strictly outside this
level: for $\|w\| \le W_{\max}$,
\begin{equation}
V(e) > \bar V_h
\Rightarrow \|e\| > \alpha_d^{-1}(\sigma(H_{\max}))
\Rightarrow \dot V(e) < 0,
\label{eq:app_strict_decrease}
\end{equation}
so $\Omega_h := \{e : V(e) \le \bar V_h\}$ is forward invariant and entered in
finite time, with ultimate bound
$\bar e_h(H_{\max}) = \underline\alpha^{-1}(\bar V_h)$. This is
Definition~\ref{def:iss} with $(\mathbb{T},\mathcal{B},\xi,w) = (\mathbb{R}_{\ge0},\mathcal{B}_L,e,w)$.

\subsection{General Safety Threshold}\label{app:subdefthre}
\begin{definition}[Safety Threshold Function]
\label{app:def_threshold}
For each constraint $i$ written $c_{e,i}^\top e \le d_i(v)$, define the Lyapunov
threshold
\begin{equation}
\Gamma_i(v) := \sup\{V(e) : c_{e,i}^\top e \le d_i(v)\},
\quad \Gamma(v) := \min_i \Gamma_i(v).
\label{eq:app_threshold_sup}
\end{equation}
\end{definition}
The linear--quadratic specialization $V(e) = e^\top P e$ reduces the supremum to
the closed form $\Gamma_i(v) = d_i(v)^2/(c_{e,i}^\top P^{-1} c_{e,i})$
of~\eqref{eq:Gamma_i}.

\subsection{General Robust Invariance}
\begin{theorem}[Robust Invariance of ERG, general]
\label{app:thm_invariance}
Under Assumptions~\ref{app:asm_frozen} and~\ref{asm:adm_dist}, the set
$\tilde K = \{(e,v) : V(e) \le \Gamma(v)\}$ is robustly forward invariant for the
composite dynamics~\eqref{eq:app_err_dyn},~\eqref{eq:erg_dynamics} for all
$\|w\| \le W_{\max}$ and any commanded reference $r(\cdot)$.
\end{theorem}

\begin{proof}
Define the barrier $\Phi(e,v) := V(e) - \Gamma(v)$ and consider
$\partial\tilde K = \{\Phi = 0\}$, i.e.\ $V(e) = \Gamma(v)$. There
$\Delta(e,v) = \max(0, \Gamma(v) - V(e)) = 0$, so the ERG
law~\eqref{eq:erg_dynamics} gives $\dot v = 0$ and the error obeys the frozen
dynamics. Since $V(e) = \Gamma(v) > \bar V_h$ by Assumption~\ref{asm:adm_dist},
the state is outside $\Omega_h$, so by~\eqref{eq:app_strict_decrease}
$\dot V(e) < 0$. With $\dot v = 0$, $\Gamma(v)$ is constant, $\dot\Phi = \dot V - \tfrac{d}{dt}\Gamma(v) = \dot V < 0$ on $\{\Phi = 0\}$. By
the invariance property of zeroing barrier functions~\cite[Prop.~1]{Ames2017},
the sublevel set $\{\Phi \le 0\} = \tilde K$ is robustly forward invariant.
\end{proof}

\begin{corollary}[Low-Level Contract, general]
\label{app:cor_lowlevel}
Under Assumptions~\ref{app:asm_frozen} and~\ref{asm:adm_dist}, $C_{\mathrm{tss}}$,
and $(e(kT_s),v(kT_s))\in\tilde K$, the low-level contract holds with
\begin{equation}
\varepsilon_L := h_r(\bar e_h(H_{\max})) = h_r\big(\underline\alpha^{-1}(\bar V_h(H_{\max}))\big).
\label{eq:app_eps_L_general}
\end{equation}
\end{corollary}

\begin{proof}
\emph{(Safety)} By Theorem~\ref{app:thm_invariance}, $V(e) \le \Gamma(v) \le
\Gamma_i(v)$ for every $i$, so $c_{e,i}^\top e \le d_i(v)$, giving
$x(t)\in\mathcal{X}_{\mathrm{safe}}$ and $u(t)\in\mathcal{U}$ (requires
$A_{\mathrm{env}}$, not $A_{\mathrm{ref}}$). \emph{(Tracking)} Once $v(t) = r$,
the frozen dynamics drive the error into $\Omega_h$ within $\tau_{LL}$; the
ultimate bound~\eqref{eq:app_strict_decrease} gives
$\|h_r(x(t_0+\tau_{LL})) - r\| \le \varepsilon_L = \bar e_h(H_{\max})$.
\end{proof}

\subsection{Recovery of the Linear--Quadratic Case}
The linear--quadratic specialization of Section~\ref{subsec:safety_layer} is
recovered by setting $V(e) = e^\top P e$ with
$\underline\alpha(s) = \lambda_{\min}(P)s^2$ and
$\overline\alpha(s) = \lambda_{\max}(P)s^2$: then~\eqref{eq:app_Vbar_general}
becomes $\bar V_h = \lambda_{\max}(P)(\alpha^{-1}(\sigma(H_{\max})))^2$
of~\eqref{eq:Vbar_def}, the sup-form~\eqref{eq:app_threshold_sup} closes to
$\Gamma_i(v) = d_i(v)^2/(c_{e,i}^\top P^{-1}c_{e,i})$ of~\eqref{eq:Gamma_i}, and
the component-wise bound~\eqref{eq:eps_L} follows from the ellipsoidal
Cauchy--Schwarz projection.

\end{document}